\documentclass[twocolumn]{IEEEtran}

\usepackage{lipsum}
\usepackage{mathtools}
\usepackage{cuted}
\usepackage[usenames]{color}
\usepackage{url}
\usepackage{cite}
\usepackage{subfigure}
\usepackage{url}
\usepackage{amsmath}
\usepackage{amsfonts}
\usepackage{algorithm}
\usepackage{array}
\usepackage[noend]{algpseudocode}
\usepackage{amsthm}
\usepackage{bm}
\usepackage{graphicx}

\usepackage{booktabs} 

\newtheorem{assumption}{Assumption}
\newtheorem{proposition}{Proposition}

\begin{document}

\title{Optimal Spectrum Sensing Policy with Traffic Classification in RF-Powered CRNs}

\author{\IEEEauthorblockN{Hae Sol Lee, Muhammad Ejaz Ahmed, and Dong In Kim} }

\maketitle

\begin{abstract}
An orthogonal frequency division multiple access (OFDMA)-based primary user (PU) network is considered, which provides different spectral access/energy harvesting opportunities in RF-powered cognitive radio networks (CRNs). In this scenario, we propose an optimal spectrum sensing policy for opportunistic spectrum access/energy harvesting under both the PU collision and energy causality constraints. PU subchannels can have different traffic patterns and exhibit distinct idle/busy frequencies, due to which the spectral access/energy harvesting opportunities are application specific.
Secondary user (SU) collects traffic pattern information through observation of the PU subchannels and classifies the idle/busy period statistics for each subchannel. Based on the statistics, we invoke stochastic models for evaluating SU capacity by which the energy detection threshold for spectrum sensing can be adjusted with higher sensing accuracy. To this end, we employ the Markov decision process (MDP) model obtained by quantizing the amount of SU battery and the duty cycle model obtained by the ratio of average harvested energy and energy consumption rates. We demonstrate the effectiveness of the proposed stochastic models through comparison with the optimal one obtained from an exhaustive method. 
\end{abstract}


\IEEEpeerreviewmaketitle

\section{Introduction}

Recently, there have been many proposals on designing efficient circuits and devices for radio frequency (RF) energy harvesting suitable for low-power wireless applications \cite{T11} -\cite{X16}. The maximum power available for RF energy harvesting at a free space distance of 40 meters is known to be 7$\mu$W and 1$\mu$W for 2.4GHz and 900MHz frequency, respectively \cite{A12}. With the RF energy harvesting capability, wireless, especially mobile device can operate perpetually without periodic energy replenishment. In cognitive radio networks (CRNs), secondary users (SUs) equipped with RF energy harvesting device can opportunistically not only access primary user (PU) channels but also harvest RF energy carried on PU channels through spectrum sensing. Hence, selecting PU channel for harvesting or transmitting through accurate spectrum sensing is a crucial component for SUs to achieve an optimal performance. 

Most works on RF energy harvesting in CRNs rely on a predefined assumption on PU channel idle time distribution \cite{L13}  -\cite {AS12} and focus on optimizing SU spectral access based on the battery level. Exploiting multiple channels will offer different idle channel distributions and PU's signal strengths. Therefore, it will give SU more chances to choose between transmitting data and harvesting energy, which in turn improves the transmission and energy harvesting efficiency, as demonstrated in \cite{D14} and \cite{D15}. However, in practice, the channel idle time distribution depends on specific traffic patterns carried over the PU channel \cite{M14}. Therefore, it is of paramount importance for SU to be aware of PU traffic patterns so that it can adapt its harvesting/transmission strategies accordingly. The challenge, however, lies in efficient classification of PU traffic patterns based on the applications' fingerprints (traffic features).

Existing solutions for the traffic patterns identification fall into the following three categories: 1) port-based, 2) signature-based, and 3) deep packet inspection. But these approaches do not perform well under the dynamic nature of the traffic patterns. In this paper, we propose Dirichlet process mixture model (DPMM) to efficiently classify various applications (traffic patterns). The DP is a family of Bayesian nonparametric (BNP) models which are mostly used for density estimation, clustering, model selection/averaging, etc. The DPs are nonparametric which means that the number of hidden traffic applications is not known in advance. Due to the nonparametric nature of these models, they do not require the number of clusters (applications) to be known {\it a priori}. Moreover, such models can adapt dynamically over time as the number of traffic patterns grows. The proposed DPMM traffic classification is unsupervised, and based only on observations without any control overhead. 

Based on the classification of PU traffic patterns, we propose an optimal spectrum sensing policy for opportunistic spectral access/energy harvesting in RF-powered CRNs. Towards this, SU collects the information on distinct traffic patterns through observation of the PU subchannels and identifies the PU traffic patterns by classifying the received data packets into distinct features. Following the DPMM approach developed in \cite{Y16}, \cite{K16}, we can classify the PU subchannels with idle/busy period statistics which will be used for optimal spectral access/energy harvesting. For this, we need to obtain the appropriate energy detection thresholds associated with distinct traffic patterns, so as to establish the optimal spectrum sensing policy with higher sensing accuracy. 

Suppose the energy detection threshold is low, then SU will likely identify the PU subchannel as busy due to noise/co-channel interference. Hence, the probability of the PU subchannel being recognized as idle is low, even if the PU subchannel in fact is idle, resulting in less transmission opportunity for SU. On the other hand, if the threshold is high, the probability of the PU subchannel being recognized as busy is low, causing SU to transmit aggressively. This will result in collision of PU and SU transmissions, reducing their transmission efficiency, and also incur the energy depletion of SU device due to its frequent transmissions. Therefore, we invoke two stochastic models for evaluating SU capacity, and then derive an optimal energy detection threshold for spectrum sensing to maximize the SU capacity. To this end, we employ the Markov decision process (MDP) model obtained by quantizing the amount of SU battery and the duty cycle model obtained by considering the average energy harvesting and energy consuming rates. We confirm that the SU capacity obtained from two stochastic models can be close to the optimal capacity obtained from an exhaustive method. 

The rest of the paper is organized as follows. Section II describes the system model along with the key assumptions. In Section III, we give an overview of traffic classification based on the DPMM approach. In Section IV, we formulate an optimization problem using the duty cycle model for deriving an optimal energy detection threshold, while in Section V we propose a stochastic model based on the MDP for SU with some statistic information about PU subchannels. Finally, the performance obtained from analysis is examined in Section VI, and concluding remarks are given in Section VII. 

\section{System Model}

We consider a RF-powered CRN as shown in Fig. 1, where SU is equipped with RF energy harvesting capability and performs opportunistic transmission or energy harvesting by accessing the corresponding PU subchannel. We assume a PU network which employs orthogonal frequency division multiple access (OFDMA) with a total of $N_{c}$ subchannels and synchronous time-slot based communication across PUs. Here, the frequency band is divided into several non-overlapping narrow frequency subbands assigned to different PUs. 

\begin{figure}[h!]
\centering
\includegraphics[width=3.2in]{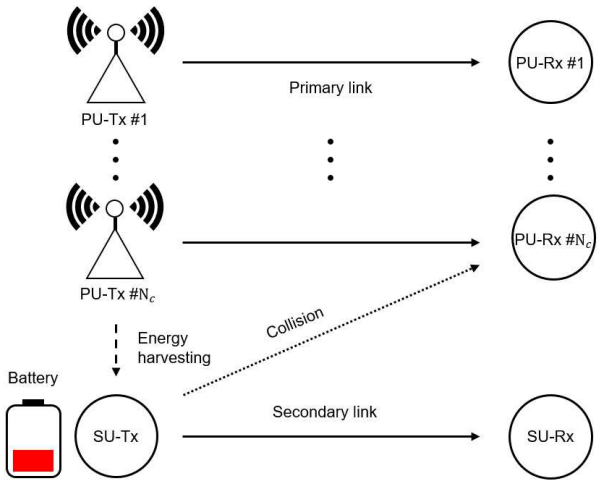}
\caption{Illustration of cognitive radio network with energy harvesting.}
\label{fig:1}
\end{figure}

We assume the subchannel of each PU shows independent idle/busy time statistics, varying with $K$ traffic sources. If SU can identify $K$ traffic sources on PU subchannels, it will increase opportunities for SU transmission and energy harvesting. To identify $K$ traffic sources and classify their traffic patterns, we consider the three features, such as the packet length, packet interarrival time, and variance in packet length. The three features are observed by SUs by inspecting the packet header from PU traffic. It is assumed that SUs collaborate with each other in sharing these features via a common control channel. 

Traffic classification is important in the problem under consideration. Each application follows a unique pattern, and recognizing this pattern is important for the following two reasons: 1) we can predict how frequently the energy arrives (packet interarrival). 2) we can predict for how long the subchannel is occupied. For each traffic application those values are different, so estimating those values would save us energy by avoiding undue spectrum sensing because in that case we do not know the pattern of the traffic application. If traffic classification is removed, we would waste undue energy by blindly sensing channels. In our approach we can use this energy to transmit data instead.

We assume the packet arrival rate on subchannel $c$ follows a random process with mean $\lambda_c$. After clustering subchannels, SU identifies the \textit{harvest} subchannel, denoted by $c_{h}$, for energy harvesting such that $\lambda_{c_{h}}$ is maximum. Similarly, the \textit{transmit} subchannel, denoted by $c_t$, is identified for transmission such that $\lambda_{c_t}$ is minimum. From the subchannel with maximum $\lambda_{c_{h}}$, SU can harvest more energy as $c_{h}$ is the subchannel with frequent energy/packet arrivals. Similarly, for the subchannel with lowest packet arrivals, i.e., $c_t$, a possible PU collision is less likely, compared to the other subchannels. We define the current state (idle/busy) of the selected subchannels for harvesting and transmission, i.e., $S_{c_{h}} \in \left \{0(idle),1(busy)\right \}$, $S_{c_{t}} \in \left \{0(idle),1(busy)\right \}$. Both subchannels coexist unless all $N_{c}$ subchannels have the same traffic pattern. If they all have the same traffic pattern and have the same subchannel states, the two subchannels are randomly selected. Given the subchannels $c_{h}$ and $c_{t}$, we define $\Pr(S_{c_{h}}=0)=p^{c_{h}}_{i}$ and $\Pr(S_{c_{t}}=0)=p^{c_{t}}_{i}$ as the probabilities of the corresponding subchannels for harvesting and transmission being idle. 

Note that the energy harvesting can be performed over multiple PU subchannels, and the channel gain affects largely the performance of harvesting and transmission, both of which need to be addressed further. Moreover, the energy harvesting using cooperation among multiple PU subchannels can increase the rate capacity of SU, especially under the {\it energy causality} in the RF-powered CRN considered herein. Due to the complexity of determining the optimal energy detection threshold through the traffic classification, thereby the optimal sensing policy, if consider such multi-channel cooperation, in this paper we demonstrate an improvement in the rate capacity by choosing the best subchannel either for energy harvesting or for transmission. Such multi-channel cooperation will be an interesting issue for further extending the framework developed in this paper to improve the rate capacity of SU with self-powering. 

\subsection{SU Battery Model}

The idle and busy probabilities on each PU subchannel for harvesting and transmission can be estimated through the PU traffic patterns identification. Here, we assume that SU battery is charged by the energy harvesting which stores energy into a rechargeable battery of finite capacity $B_{max} \in \mathbb{R}^{+}$ for a non-negative real number $\mathbb{R}^{+}$. As shown in Fig. 2 below, SU selects active mode or sleep mode for which the action can be denoted as $a_{t} \in \left \{ 0(sleep),1(active)\right \}$ in the slot $t$. 

\begin{figure}[h!]
\centering
\includegraphics[width=3.5in]{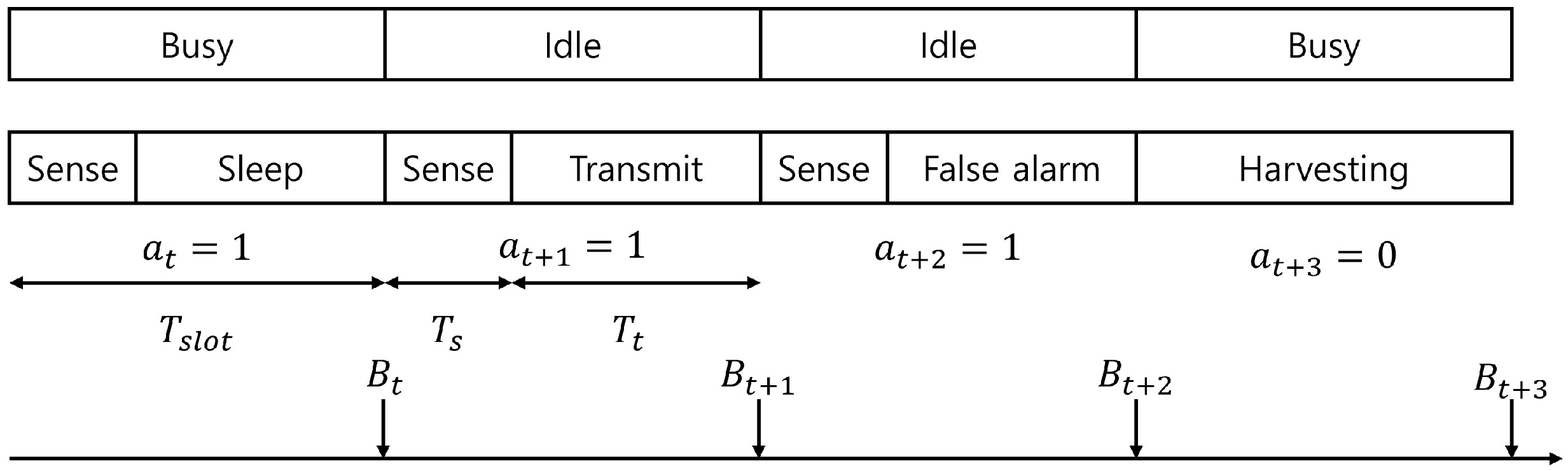}
\caption{The frame structure of SU with energy harvesting. $a_t$ represents the mode selection of SU in a slot t. }
\label{fig:2}
\end{figure}

\noindent
In Fig. 2, we denote a time slot by $T_{slot}$ where each $T_{slot}$ is divided into sensing and transmission time durations. We represent the sensing time duration with $T_{s}$, the transmission time duration with $T_{t} =T_{slot} - T_{s}$, and the residual battery level in a slot $t$ with $B_{t}$. For simplicity, we assume that SU always has data to transmit. 

SU selects the mode (active or sleep) according to the battery level observed in each time slot. Then, by comparing the residual battery level $B_{t}$ at current slot $t$ with sufficient energy for transmission, we define the action as 
\begin{equation}
\label{eqn:Ei}
a_{t}=
\begin{cases}
1, & B_{t} \geq e_{s}+e_{t}\\ 
0, & B_{t} < e_{s}+e_{t}
\end{cases}
\end{equation}
where $e_{s}$ is the energy consumption for sensing and $e_{t}$ the energy consumption for transmission. 

In active mode, SU performs spectrum sensing with the energy $e_{s} = T_{s}P_{s}$ consumed over $T_{s}$ with sensing power $P_{s}$. We represent the sensing outcome as observations, i.e., $o_{t} \in \left \{0(idle),1(busy)\right \}$. If SU observation is busy ($o_{t}=1$) after sensing the subchannel $c_{t}$, SU does not transmit data in the transmission phase ($T_{t}$) of a slot. On the contrary, if $o_{t}=0$, SU transmits data that consumes $e_{t}=T_{t}(P_{t}/\eta +P_{nc})$ in the transmission phase where $P_t$ is the SU transmit power, $P_{nc}$ the non-ideal circuit power, and $\eta\,(0<\eta\leq 1)$ the efficiency of power amplifier \cite{J14}. Note that if $o_{t}=1$ in active mode, SU can still harvest energy from the subchannel $c_{h}$. But we do not account for this as clustering can secure $c_{t}$ with higher $p^{c_{t}}_{i}$, which can be ignored here for tractable analysis. On the other hand, in sleep mode, SU turns off its transceiver (except an energy harvesting circuit) until next slot begins. 

To find the amount of residual energy in battery in each time slot, the harvested energy $E^{h}_{t}$ and energy consumption $E^{c}_{t}$ in the slot $t$ are defined as 
\begin{align}
E^{h}_{t}&=g(1-a_{t})\varphi P_{p}T_{t}S_{c_{h}} 
\\
E^{c}_{t}&=a_{t}\left [ e_{s}+(1-o_{t})e_{t} \right ] 
\end{align}
where $P_{p}$ is the PU transmit power and $g$ is the subchannel gain between PU and SU. Here we assume that SU receives the subchannel gain information via a common control subchannel. Therefore, the residual energy in battery in the next slot $t+1$ can be updated as 
\begin{equation}
B_{t+1}=\min(B_{t}-E^{c}_{t}+E^{h}_{t},B_{max}). 
\label{res}
\end{equation}

\subsection{Spectrum Sensing and Duty Cycle Model}

The sensing probability varies with specific energy detection threshold and subchannel condition associated with PU subchannel. The PU signal and noise are assumed to be modeled as independent circularly symmetric complex Gaussian (CSCG) random processes with mean zero and variances $\sigma^{2}_{p}$ and $\sigma^{2}_{w}$, respectively. Then, from \cite{W14}, the probabilities of false alarm $P_{f}(\epsilon )$ and detection $P_{d}(\epsilon)$ for the transmission subchannel $c_{t}$ are evaluated as 
\begin{align}
P_{f}(\epsilon )&=\Pr(o_{t}=1|S_{c_{t}}=0,a_{t}=1) 
\nonumber \\ 
&=Q\left [ \left ( \frac{\epsilon}{\sigma^{2}_{w}}-1 \right )\sqrt{N_{s}} \!\:\right ]
\\
P_{d}(\epsilon | g)&=\Pr(o_{t}=1|S_{c_{t}}=1,a_{t}=1,g) 
\nonumber \\ 
&=Q\left [ \left ( \frac{\epsilon}{(\frac{g\sigma^{2}_{p}}{\sigma^{2}_{w}}+1)\sigma^{2}_{w}}-1 \right )\sqrt{N_{s}} \!\:\right ] 
\end{align}
where $\epsilon \in \mathbb{R}^{+}$ is a detection threshold for the energy detector, $N_{s}$ denotes the number of samples, and $Q(x)$ Q-function.

In order to analyze a stochastic performance of SU, we need to find how often the current battery level has enough energy to transmit, i.e., the probability of SU being active $\Pr(a_{t} = 1)$. Note that we take into account the available energy to formulate the SU's action policy in (\ref{pa}) below. For this we consider the duty-cycling behavior between active mode and sleep mode to formulate $\Pr(a_{t} = 1)$. Thus, we define $\Pr(a_{t} = 1) = T_{active}/(T_{active} + T_{sleep})$, where $T_{active}$ and $T_{sleep}$ are the average times spent in active and sleep modes, respectively. Assuming that the energy harvested in sleep mode must equal the energy consumed during active mode, the probability of SU being active can be expressed as
\begin{eqnarray}
P_{a}(\epsilon | g) = Pr(a_{t} = 1) = \frac{\rho_{h}}{\rho_{h}+\rho_{c}(\epsilon | g)}
\label{pa}
\end{eqnarray}
\begin{eqnarray}
\rho_{h}=g\!\: \xi\!\: p_{o}^{c_{h}}
\end{eqnarray}
\begin{eqnarray}
\rho_{c}(\epsilon | g) = e_{s}+\Big\{\big[ 1-P_{f}(\epsilon)\big] p_{i}^{c_{t}}+\big[ 1-P_{d}(\epsilon | g)\big] p_{o}^{c_{t}}\Big\}\!\: e_{t}
\end{eqnarray}
where $\xi=\varphi P_{p}T_{t}$ and $p_{o}^{c}$ is the probability of the subchannel $c\in\{ c_{h},c_{t}\}$ being busy with $p_{o}^{c}=1-p_{i}^{c}$. In the above, we have assumed that during sleep mode, SU consumes no energy but harvests the energy with rate $\rho_{h}$. On the other hand, during active mode, SU consumes the energy at rate $\rho_{c}$. Hence, the probability of SU being active based on the duty cycle model above represents the ratio of the energy harvesting rate to the sum of the energy harvesting and energy consuming rates. 

\section{Traffic Classification}

Sensing multiple PU subchannels may be required to identify idle/busy subchannels among them, which carry $K$ traffic applications (patterns), without any prior information. This may not lead to a performance gain of SU under the {\it energy causality} in the RF-powered CRN because of undue sensing time and high energy consumption costs. To overcome this, we attempt to classify traffic patterns with observed features and select the best subchannel either for energy harvesting or for transmission accordingly. In this paper, we consider the following three traffic features with regard to application fingerprinting: packet length ($p_l$), packet interarrival time ($p_{t}$), and variance in packet length ($\Delta$), which are sufficient for acquiring the channel state information, as detailed in \cite{M12} and \cite{ME14}. Then, the packet length vector $\mathbf{P}_{\text{len}}$ is represented as $\mathbf{P}_{\text{len}}=\{p_{l_1}, p_{l_2}, \cdots, p_{l_N} \}$, where $N$ is the number of packet length samples. Similarly, $\mathbf{P}_{\text{int}}$ and $\boldsymbol{\Delta}$ are represented as $\mathbf{P}_{\text{int}}=\{p_{t_1}, p_{t_2}, \cdots, p_{t_N}\}$ and $\boldsymbol{\Delta}=\{[var(p_l)]_{W_1}=0, [var(p_l)]_{W_2}, \cdots, [var(p_l)]_{W_{N}}\}$, respectively. Here, $[var(p_l)]_{W_n}$ is the {\it temporal} variance of packet length in a window $W_n$ of size $n$, spanning over $[1,\ldots,n]$(st)th observations. 

With regard to some application backgrounds, we provide more detailed description of each feature as follows: 
\begin{enumerate}
\item Packet Length: Packet lengths for different traffic payloads are likely to
be different. For example, the UDP packet size is longer, the gaming packet size may vary depending on game dynamics, and the VoIP packet has smaller packet lengths to minimize jitter. So we use the packet length as one feature point for identifying different traffic applications.
\item Packets Interarrival Time: Packet interarrival times for different applications also vary depending upon the requirements of applications. For example, in VoIP, the inter-arrival time is small to avoid annoying effects caused by jitter.
\item Variance in Packet Length: The packet lengths may change in every connection. For example, by investigating real wireless traces in \cite{O1}, \cite{O2}, for gaming data, we have observed that packet lengths vary significantly during the communication. 
\end{enumerate}
Let $\mathbf{x}_n$ be the feature vector of the $n$th training feature point, given by
\begin{align}\label{eqn:obsr1}
\mathbf{x}_n = \!~\big[p_{l_n}, p_{t_n}, [var(p_l)]_{W_n} \big], \nonumber \quad \mathbf{X} = \big[\mathbf{x}_1,\cdots,\mathbf{x}_N\big]^T.
\end{align}
The matrix of observations $\mathbf{X}$ is observed by SU by examining the packet header of PU. 

\begin{figure}[h!]
\centering
\includegraphics[width=2.2in]{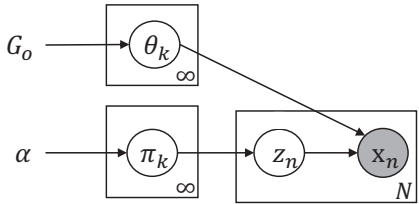}
\caption{Graphical model of an DPMM. Observations are represented by the shaded node.
Nodes are random variables, edges are dependencies, and plates are replications.}
\label{fig:SM_clusters}
\end{figure}

In practice, we do not know how many traffic applications are active in the PU network, and therefore, we would like to learn it from the data observed. The Dirichlet processes can be employed to have a mixture model (DPMM) with infinite components (applications) which can be viewed as taking the limit of the finite mixture model for $K \rightarrow \infty$.
The theoretical model of data generation for the DPMM is
\begin{equation}
\begin{aligned}
\mathbf{x}_n \quad &\sim \quad p(\theta_k) \\
\theta_k \quad &\sim \quad G \\
G \quad &\sim \quad  DP(\alpha, G_o)
\end{aligned}
\label{eqn:dp}
\end{equation}
where $G=\sum_{k=1}^{\infty}\pi_k\delta_{\theta_k^*} \sim DP (\alpha,G_o)$, $G_o$ is the base distribution and $\delta_{\theta_k^*}$ is used as a short notation for $\delta(\theta=\theta_k^*)$ which is a delta function that takes 1 if $\theta=\theta_k^*$ and 0 otherwise. $\theta_k$ are the cluster parameters sampled from $G$ where $k \in \{1,2,\dots\}$. The generative distribution $p(\theta_k)$ is governed by the cluster parameters $\theta_k$ and is used to generate the observations $\mathbf{x}_n$. Then, the multimodal probability distribution can be defined as $p(\mathbf{x}_n)=\sum_{k=1}^{\infty}\pi_k\;p(\cdot|\delta_{\theta_k^*})$ which is called the mixture distribution with mixing weights $\pi_k$ and mixing components $p(\cdot|\delta_{\theta_k^*})$. The graphical model of the DPMM is shown in Fig. \ref{fig:SM_clusters}. $\alpha$ is the scalar hyperparameter of the DPMM and affects the number of clusters obtained. $Z_n$ is the cluster assignment variable such that the feature point $\mathbf{x}_n$ belongs to the $k$th cluster. Larger the value of $\alpha$, the more clusters, while smaller the value of $\alpha$, the fewer clusters. Note that the value of $\alpha$ indicates the strength of belief in $G_o$. A large value means that most of the samples will be distinct and have values concentrated on $G_o$. 

\subsection{DPMM Representation}

The model defined above is a theoretical one. In order to realize DPMM models, Stick-breaking process, Chinese restaurant process, and Polya-urn process are used. Here we represent the DPMM in (\ref{eqn:dp}) using the Stick-breaking process. Consider two infinite collections of independent random variables $V_k \overset{i.i.d.}{\sim} \text{Beta}(1,\alpha)$ and $\theta_k^* \overset{i.i.d.}{\sim} G_o$ for $k=\{1,2,\dots\}$. The Stick-breaking process of $G$ is given by
\begin{align}\label{eqn:stickbreak}
\pi_k &= V_k\prod_{j=1}^{k-1}(1-V_j) \\
G &= \sum_{k=1}^{\infty}\pi_k \delta_{\theta_k^*} 
\end{align}
where $\mathbf{V}=\{V_{1},V_{2},\ldots\}$ with $V_{0}=0$. The mixing weights $\{\pi_k\}$ are given by breaking a stick of unit length into infinitely small segments. In the DPMM, the vector $\boldsymbol{\pi}$ represents the infinite vector of mixing weights and $\{\theta_1^*, \theta_2^*, \dots\}$ are the atoms which correspond to mixing components. Since $Z_n$ is the cluster assignment random variable to the feature point $\mathbf{x}_n$, the data for the DPMM is generated as
\begin{enumerate}
\item{Draw $V_k|\alpha \overset{i.i.d.}{\sim} \text{Beta}(1,\alpha)$, \quad $k \in \{ 1,2, \dots\}$}.
\item{Draw $\theta_k^* \overset{i.i.d.}{\sim} G_0$, \quad $k \in \{ 1,2, \dots\}$}.
\item{For the feature point $\mathbf{x}_n$, do:}
\begin{itemize}
\item[(a)]{Draw $Z_n|\{ v_1,v_2,\dots \} \; \sim \; \text{Multinomial}(\boldsymbol{\pi})$}.
\item[(b)]{Draw $\mathbf{x}_n|z_n \; \sim \; p(\mathbf{x}_n|\theta_{z_n}^*)$}.
\end{itemize}
\end{enumerate}
Here we restrict ourselves to the DPMM for which the observable data are drawn from Normal distribution and where the base distribution for the DP is the corresponding conjugate distribution.

\subsection{Inference for DPMM}

Since the Dirichlet processes are nonparametric, we cannot use the EM algorithm to estimate the random variables $\{Z_n\}$ (which store the cluster assignments) for our DPMM model in (\ref{eqn:dp}) due to the fact that EM is generally used for inference in a mixture model, but here $G$ is nonparametric, making EM difficult. Hence, in order to estimate these assignment variables in the paradigm of Bayesian nonparametrics, there exist two candidate approaches for inferences: First, a sampling-based approach uses Markov chain Monte Carlo (MCMC) to cluster the traffic patterns $\mathbf{X}$. The MCMC based sampling (also known as Gibbs sampling) approach is more accurate in classifying feature points. 
The second approach is based on variational methods which convert inference problems into optimization problems \cite{K16}. The main idea that governs variational inference is that it formulates the computation of marginal or conditional probability in terms of an optimization problem that depends on a number of free parameters, i.e., variational parameters. 
We discuss both approaches in the following subsections with their detailed comparison in Table \ref{table1}, which are confirmed with the simulation results. 

\subsubsection{Collapsed Gibbs Sampling for Traffic Classification}

In our model, the data follow multivariate normal distribution, $\mathcal{N}_k(\vec{\mu}_k,\Sigma_k)$, where the parameters are 3-dimensional mean vector and covariance matrix. The conjugate distributions for mean vector $\vec{\mu}_k$ and covariance matrix $\Sigma_k$ are given by $\vec{\mu}_k \sim \mathcal{N}(\vec{\mu}_0,\Sigma_k/\kappa_0)$ and $\Sigma_k \sim \text{Inverse-Wishart}_{\nu_0} (\Lambda_0^{-1})$, respectively.
The Dirichlet hyperparameters, here symmetric $\alpha/K$, encodes our beliefs about how uniform/skewed the class mixture weights are. The parameters to the Normal times Inverse-Wishart prior, $\Lambda_0^{-1}, \nu_0, \kappa_0$ imply our prior knowledge regarding the shape and position of the mixture densities. For instance, $\vec{\mu}_0$ specifies where we believe the mean of the mixture densities are expected to be, where $\kappa_0$ is the number of pseudo-observations we are willing to ascribe to our belief. The hyper-parameters $\Lambda_0^{-1},\nu_0$ behave similarly for the mixture density covariance.

Collapsed Gibbs sampler requires to select the base distribution $G_o$ which is a conjugate prior of the generative distribution $p(\mathbf{x}_n|\theta_{z_n}^*)$, in order to solve analytically and be able to sample directly from $p(Z_n|Z_{-n},\mathbf{X})$. The posterior distribution under our model is
\begin{eqnarray}\label{eqn:posterior}
P(Z,\Theta,\boldsymbol{\pi}, \alpha) & \hspace{-0.1in} \propto \hspace{-0.1in} & P(\mathbf{X}|Z,\Theta)P(\Theta|G_0) \left(\prod_{n=1}^{N}P(z_n|\boldsymbol{\pi})\right) \nonumber \\ 
& & \cdot\: P(\boldsymbol{\pi}|\alpha)P(\alpha). 
\end{eqnarray}
By integrating-out certain parameters, the posterior distribution is given by \cite{Wood} 
\begin{equation}\label{eqn:posterior}
P(z_i=k|Z_{-n},\mathbf{X},\Theta,\boldsymbol{\pi},\alpha) \propto P(\mathbf{x}_n|z_n,\Theta)P(z_n|Z_{-n},\alpha). 
\end{equation}
For the first term in the above equation, we use multivariate Student-t distribution, i.e., $t_{\nu_n-2}\Big[\vec{\mu}_n,\frac{\Lambda_n(\kappa_n+1)}{\kappa_n(\nu_n-2)}\Big]$, since we chose Inverse-Wishart as conjugate prior for $\Sigma_k$ and Normal distribution for $\vec{\mu}_k$, where the second term is called Chinese restaurant process and is given by 
\begin{equation}\label{eqn:crp}
P(z_n=k|Z_{-n})=\left\{
\begin{array} {l l}
\displaystyle \frac{m_k}{n-1+\alpha}, \text{  if } k \leq K_+, \\
\displaystyle \frac{\alpha}{n-1+\alpha}, \text{  if } k > K_+
\end{array} \right.
\end{equation}
where $Z_{-n}=Z/z_{n}$, $K^+$ is the number of classes containing at least one data point, and $m_k=\sum_{n=1}^{N}I(z_i=k)$ is the number of data points in class $k$. 

The steps involved in collapsed Gibbs sampling are enumerated below as:
\begin{enumerate}
\item{Initialize the cluster assignments $\{ z_n\}$ randomly}.
\item{Repeat until convergence:}
\begin{itemize}
\item[(a)]{Randomly select $\mathbf{x}_n$}.
\item[(b)]{Fix all other $z_n$ for every $n \neq n$: $Z_{-n}$}.
\item[(c)]{Sample $z_n \sim p(Z_{n}|Z_{-n},\mathbf{X})$ from (\ref{eqn:crp})}.
\item[(d)]{If $z_n > K$, then update $K=K+1$}.
\end{itemize}
\end{enumerate}

An overall framework for PU traffic pattern classification and application-specific optimization is shown in Fig. \ref{fig:framework}.

\begin{figure}[h!]
\centering
\includegraphics[width=3in, height=5in]{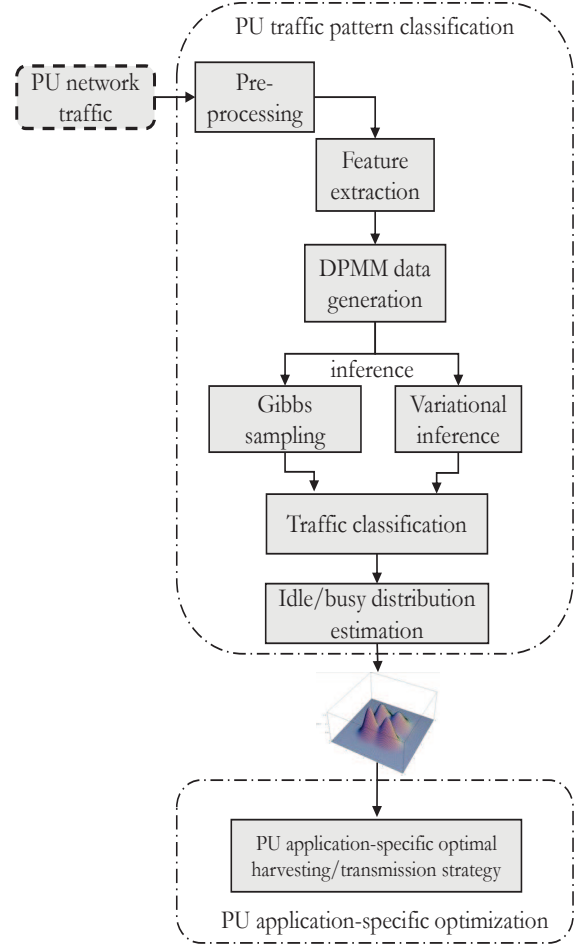}
\caption{A flow diagram for classification of PU traffic patterns.}
\label{fig:framework}
\end{figure}

\subsubsection{Variational Bayes Inference for Traffic Classification}

The main idea that governs variational inference is that it formulates the computation of marginal or conditional probability in terms of an optimization problem that depends on the number of free parameters, i.e., variational parameters. In other words, we choose a family of distributions over the latent variables with its own set of variational parameters $\nu$, i.e., $q(W_{1:M}|\nu)$. We are interested in finding the setting of the parameters that makes our approximation $q$ closest to the posterior distribution. Consequently, we can use $q$ with the fitted parameters in place of the posterior.

Here we assume $N$ observations, $\mathbf{X}=\{\mathbf{x}_{n}\}_{n=1}^{N}$ and $M$ latent variables, i.e., $\mathbf{W}=\{W_{m}\}_{m=1}^{M}$. The fixed parameters $\mathcal{H}$ could be the parametrization of the distribution over the observations or the latent variables. With the given notations, the posterior distribution under Bayesian paradigm is
\begin{equation}\label{eqn:bayesian}
p(\mathbf{W}|\mathbf{X},\mathcal{H})=\frac{p(\mathbf{W},\mathbf{X}|\mathcal{H})}{\int_{\mathbf{W}}p(\mathbf{W},\mathbf{X}|\mathcal{H})}.
\end{equation}
The posterior density in (\ref{eqn:bayesian}) is in an intractable form (often involving integrals) which cannot easily be solved analytically. Thus we rely on an approximate inference method, i.e., variational Bayes inference. Our goal is to find an approximation of the posterior distribution $p(\mathbf{W}|\mathbf{X},\mathcal{H})$ as well as the model evidence $p(\mathbf{X}|\mathcal{H})$. We introduce a distribution $q(\mathbf{W})$ defined over the latent variables and observe that for any choice of $q(\mathbf{W})$, the following decomposition holds
\begin{equation}\label{eqn:posterior}
\log p(\mathbf{X}|\mathcal{H}) = \mathcal{L}(q)+\text{KL}(q||p)
\end{equation}
where
\begin{align}
\mathcal{L}(q) &= \int q(\mathbf{W}) \log \Bigg\{\frac{p(\mathbf{X},\mathbf{W}|\mathcal{H})}{q(\mathbf{W})}\Bigg\}d\mathbf{W} \label{eqn:lb} \\
\text{KL}(q||p) &= -\int q(\mathbf{W}) \log \Bigg\{\frac{p(\mathbf{W}|\mathbf{X},\mathcal{H})}{q(\mathbf{W})}\Bigg\}d\mathbf{W}. \label{eqn:kl}
\end{align}
From (\ref{eqn:kl}), we see that $\text{KL}(q||p)$ is the Kullback-Leibler divergence between $q(\mathbf{W})$ and the posterior distribution $p(\mathbf{W}|\mathbf{X},\mathcal{H})$, where the KL divergence satisfies $\text{KL}(q||p)\geq 0$ with equality if and only if $q(\mathbf{W})=p(\mathbf{W}|\mathbf{X},\mathcal{H})$, i.e., $q(\cdot)$ is the true posterior.
It follows from (\ref{eqn:posterior}) that $\mathcal{L}(q) \leq \log p(\mathbf{X}|\theta)$, in other words, $\mathcal{L}(q)$ is a lower bound on $p(\mathbf{X}|\theta)$. Therefore, we can maximize the lower bound $\mathcal{L}(q)$ by optimizing with respect to $q(\mathbf{W})$, which is equivalent to minimizing the KL divergence. We consider a restricted family of distributions $q(\mathbf{W})$ and then seek the member from this family for which the KL divergence is minimized.

The latent variables\footnote{We use the variational inference to pick a family of distributions over the latent variables $\mathbf{W}$ with its own variational parameters $\bm{\nu}$, and then set $\bm{\nu}$ to render ${q(\cdot|\bm{\nu})}$ close to the posterior of interest.} for DPMM are the stick lengths, atoms, and cluster assignment variables: $\mathbf{W}=\{\mathbf{V}, \bm{\theta^*}, \mathbf{Z}\}$ and the hyperparameters are the scaling parameter $\alpha$ and the parameter for the conjugate base distribution $G_0$, $\mathcal{H}=\{\alpha, G_0\}$. Thus, the marginal distribution of the data (evidence) lower bound is evaluated as 
\begin{align}\label{eqn:optimize2}
\log p(\mathbf{X}|\mathcal{H}) \geq &\:\mathbf{E}_q[\log p(\mathbf{V}|\mathcal{H})]+ \mathbf{E}_q[\log p(\bm{\theta^*}|G_0)] \nonumber \\
&+ \sum \limits_{n=1}^N (\mathbf{E}_q[\log p(Z_n|\mathbf{V})]+\mathbf{E}_q[\log p(\mathbf{x}_n|Z_n)]) \nonumber \\
&- \mathbf{E}_q[\log q(\mathbf{V},\bm{\theta^*},\mathbf{Z})].
\end{align}
To maximize the bound, we must find a family of variational distributions that approximate the distributions of infinite-dimensional random measure $G$, where $G$ is expressed in terms of $\mathbf{V}=\{V_1, V_2, \dots\}$ and $\bm{\theta^*}=\{\theta_1^*, \theta_2^*,\dots \}$. The factorized family of variational distributions for mean-field inference can be expressed as
\begin{align}\label{eqn:vparameter}
q(\mathbf{v},\bm{\theta^*},\mathbf{z})=\prod \limits_{k=1}^{K-1}q(v_k|\zeta_k)\prod \limits_{k=1}^{K}q(\theta_k^*|\varepsilon_k)\prod \limits_{n=1}^N q(z_n|\vartheta_n)
\end{align}
where $q(v_k|\zeta_k)$ are beta distributions parameterized with $\zeta_k$, $q(\theta_k^*|\varepsilon_k)$ are exponential family distributions with natural parameter $\varepsilon_k$, and $q(z_n|\vartheta_n)$ are multinomial distributions with parameter $\sigma_n$. The latent variables in Fig. \ref{fig:SM_clusters} are governed by the same distribution, but, following the fully factorized variational variables in the mean-field variational approximation, given there is an independent distribution for each variable. Thus, the variational parameters are defined by
\begin{align}\label{eqn:vparameters}
\bm{\nu} = \{\zeta_1,\dots,\zeta_{K-1}, \varepsilon_1,\dots,\varepsilon_K,\vartheta_1,\dots,\vartheta_N \}.
\end{align}
Since multiple parameter options exist for each latent variable under the variational distribution, we need to optimize the bounds in (\ref{eqn:optimize2}) based on $\bm{\nu}$ above. 

Optimizing by employing the coordinate ascent algorithm from \cite{Blei}, we optimize the  bounds in (\ref{eqn:optimize2}) with respect to the variational parameters $\bm{\nu}$ in (\ref{eqn:vparameters}). From the coordinate ascent algorithm, we acquire the following statistics: the number of traffic patterns ($K$) and their corresponding parameters $\mathbf{\theta}_k^*$. These statistics are utilized by SU for optimal harvesting and transmission strategy in the sequel. 

\subsection{Comparison of Inference Algorithms}

Table \ref{table1} presents the detailed comparison of the two algorithms. The collapsed Gibbs sampling is more accurate than the variational Bayes inference, but the former has some limitations such as slow convergence \cite{Y16} and difficult to diagnose convergence. We confirm the comparison results, in terms of accuracy and time complexity, through simulations. The variational inference is biased (underfitting), whereas the Gibbs sampling's bias diminishes as the number of runs for the Markov chain increases. For non-conjugate prior distributions, the latter is preferred which is much faster than the former. The latter is deterministic, which means that we always obtain the same optimal value, given the same starting value and an objective function without huge local optima problems. The latter is quicker since it approximates the posterior using optimization with free variables.

\begin{table}[h]
\caption{Comparison of inference algorithms.}
\label{tab:results}
\begin{tabular}{cccl}
\toprule
Item & Gibbs sampling & Variational inference \\
\midrule
Speed & Slower& Faster \\
Biasness & Not biased & Biased \\
Computational requirement & Higher & Lower \\
For non-conjugate dist. & Not preferred & Preferred \\ 
Deterministic & No & Yes \\
Accuracy & Higher & Lower \\
Convergence diagnoses & Difficult & Easy \\
Converge to true posterior & Yes & No \\
Inference approach & MCMC sampling & Optimization \\
Approximates & Integrals & Data distribution \\
\bottomrule
\end{tabular}\label{table1}
\end{table}

\section{Optimal Energy Detection Threshold Estimation}

After selecting the subchannels for harvesting/transmission, SU uses $c_{h}$ to harvest energy and $c_{t}$ to transmit information. Specifically, the goal is to adjust the detection threshold $\epsilon$ of the SU energy detector under the energy causality and PU collision constraints. For instance, increasing the detection threshold results in frequent SU transmissions, as a result, the increased probability of accessing the occupied spectrum, which may result in collision with PU transmission. Also, it incurs excessive energy usage which is not good for energy-constrained SU. On the other hand, lowering the detection threshold $\epsilon$ reduces unnecessary sensing and transmission actions and consequently saves energy for future transmission. However, it decreases the probability of accessing the unoccupied spectrum, causing the achievable rate loss to SU. Thus, it is of vital importance to finely tune the SU energy detection threshold $\epsilon$ for optimal sensing subject to the design constraints, i.e., PU protection and energy causality.

\subsection{Problem Formulation}

The achievable rate capacity of SU is defined as $\boldsymbol{C} = W\log _{2}(1+SNR)$ for the signal-to-noise ratio ($SNR$) of the SU link when the PU transmit subchannel $c_{t}$ of bandwidth $W$ is idle. Then the SU rate capacity is expressed as 
\begin{align}
R\left ( \epsilon | g \right ) &=\frac{T_{t}}{T_{slot}}\, \boldsymbol{C}\Pr(a_{t}=1,o_{t}=0,S_{c_{t}}=0  | g)\hspace{-0.3in}
\nonumber \\ 
&= \frac{T_{t}}{T_{slot}}\, \boldsymbol{C} \big[ 1-P_{f}\left ( \epsilon \right ) \big] P_{a}(\epsilon | g)\!\: p^{c_{t}}_{i}. 
\label{rate}
\end{align}
Here we have used the relation $\Pr(a_{t}=1,o_{t}=0,S_{c_{t}}=0  | g)=\Pr(o_{t}=0|a_{t}=1,S_{c_{t}}=0)\Pr(a_{t}=1|S_{c_{t}}=0,g)\Pr(S_{c_{t}}=0)$ based on Bays' rule. The rate capacity above converges to a specific value due to the Q-function characteristic of $P_{f}$ and $P_{a}$ as the energy detection threshold $\epsilon$ increases. 

In this case, however, the PU performance is degraded due to the collision of PU and SU transmissions as the latter becomes more aggressive. Therefore, we should put some constraint on the collision probability, which is evaluated as 
\begin{align}
P_{c}(\epsilon  | g)&=\Pr(a_{t}=1,o_{t}=0|S_{c_{t}}=1,g) \hspace{-0.37in}
\nonumber \\ 
&=\big[ 1-P_{d}(\epsilon | g) \big] P_{a}(\epsilon | g ). 
\label{con}
\end{align}
Now we can formulate an optimization problem to find an appropriate value of $\epsilon$, which leads to an optimal spectrum sensing policy for maximizing the SU rate capacity as 
\begin{eqnarray}
\max_{\epsilon }~ R\left ( \epsilon | g \right ) ~~\textrm{s.t.}~~ P_{c}(\epsilon | g)\leq \overline{P}_{c} 
\label{prob}
\end{eqnarray}
where $\overline{P}_{c}$ is the target probability of collision with which PU can be sufficiently protected.

\subsection{Distinctions of RF-Powered CRNs from General CRNs}

As defined in (\ref{pa}), the probability of action for SU in RF-powered CRNs is a function of its battery level, given the {\it energy causality} (i.e., self-powering) is applied to SU for joint opportunistic energy harvesting and transmission, unlike general CRNs. Therefore, as formulated in (\ref{rate}) - (\ref{prob}), the optimal sensing policy in the RF-powered CRNs should take into account the energy state, unlike that in the general CRNs. For this, we have acquired the channel state information through the traffic classification developed in this paper, which is a crucial factor for determining the optimal sensing policy in the RF-powered CRNs.

\subsection{Optimal Energy Detection Threshold}

To find an optimal energy detection threshold $\epsilon$ in (\ref{prob}), we define an objective function as $O(\epsilon | g)=\big[ 1-P_{f}\left ( \epsilon \right )\big] P_{a}(\epsilon | g)$ which is affected only by $\epsilon$ in (\ref{rate}), and then reformulate the optimization problem above as 
\begin{eqnarray}
\max_{\epsilon }~ O\left ( \epsilon  | g\right ) ~~\textrm{s.t.}~~ P_{c}(\epsilon | g)\leq \overline{P}_{c}. 
\label{obj}
\end{eqnarray}
In addition, we define the constraint function as $\Phi (\epsilon, \overline{P}_{c})=P_{c}(\epsilon | g)- \overline{P}_{c}$ to obtain a proper threshold range of $\left [ 0,\epsilon_{c} \right ]$, where $\epsilon_{c}$ is the solution 
of the following equation: 
\begin{eqnarray}
\Phi (\epsilon_{c}, \overline{P}_{c}) = 0. 
\label{sol}
\end{eqnarray}

We will find the following Propositions 1 and 2 useful in obtaining the solution $\epsilon_{c}$ of (\ref{sol}).
\begin{proposition}
The collision probability $P_{c}(\epsilon | g)$ in (\ref{con}) and the object function $O(\epsilon | g)$ in (\ref{obj}) can be classified into two types of function $f(\epsilon|g)$ as follows:
\begin{enumerate}
\item $f(\epsilon_{1}|g)<f(\epsilon_{2}|g)$ ~~for $\epsilon_{1}<\epsilon_{2} \leq \epsilon_{m}$, \\
$f(\epsilon_{3}|g)>f(\epsilon_{4}|g)$ ~~for $\epsilon_{m}\leq\epsilon_{3} < \epsilon_{4}$, \\
$\lim_{\epsilon\rightarrow \infty}f(\epsilon|g) =\gamma_{1}$
~~where
$\epsilon_{m}=\arg\max_{\epsilon}f(\epsilon|g)$.
\item $f(\epsilon_{1}|g)<f(\epsilon_{2}|g)$ ~~for $\epsilon_{1}<\epsilon_{2}$, \\
$\lim_{\epsilon\rightarrow \infty}f(\epsilon|g) =\gamma_{1}$
\end{enumerate}
where
\begin{eqnarray}
\gamma_{1}=\frac{\rho_{h}}{\rho_{h}+e_{s}+e_{t}}. 
\label{gam}
\end{eqnarray}
\begin{proof}
See Appendix A.
\end{proof}
\end{proposition}
\begin{proposition}
The constraint function in (\ref{sol}) yields a unique solution $\epsilon_{c}$ for $\gamma_{1}>\overline{P}_{c}$. If $\gamma_{1}\leq\overline{P}_{c}$, we have two or no solution. 
\end{proposition}

According to Proposition 2, we have the constraint range $[0,\epsilon_{c}]$ if $\gamma_{1}>\overline{P}_{c}$ and otherwise, the constraint range is $[0,\infty)$. Following the IEEE 802.22 WRAN, if $\overline{P}_{c}=0.1$, we consider $\gamma_{1}\leq\overline{P}_{c}$ as an extreme case corresponding to an extremely low energy harvesting rate, namely $\rho_{h}\ll (e_{s}+e_{t})$ in (\ref{gam}). Thus, we may choose the constraint range to be $[0,\epsilon_{c}]$ as the only possible option. 

To find $\epsilon_{c}$, we resort to the secant method as a root-finding algorithm where the constraint function can be approximated by a secant line through two points of the function. 
Starting with the two initial iterates $\epsilon_{0}$ and $\epsilon_{1}$, the next iterate $\epsilon_{2}$ is obtained by computing the value at which the secant line passing through the two points $(\epsilon_{0},\Phi(\epsilon_{0},\overline{P}_{c}))$ and $(\epsilon_{1},\Phi(\epsilon_{1},\overline{P}_{c}))$ as 
\begin{eqnarray}
\frac{\Phi(\epsilon _{1}, \overline{P}_{c})-\Phi(\epsilon _{0}, \overline{P}_{c})}{\epsilon _{2}-\epsilon _{1}}(\epsilon _{1}-\epsilon _{0})+\Phi(\epsilon _{1}, \overline{P}_{c})=0,
\end{eqnarray}
which yields the solution
\begin{eqnarray}
\epsilon _{2}=\epsilon _{1}-\Phi(\epsilon _{1}, \overline{P}_{c})\frac{\epsilon _{1}-\epsilon _{0}}{\Phi(\epsilon _{1}, \overline{P}_{c})-\Phi(\epsilon _{0}, \overline{P}_{c})}.
\end{eqnarray}
Hence, we can derive the recurrence relation as 
\begin{eqnarray}
\epsilon _{k}=\frac{\epsilon _{k-2}\Phi(\epsilon _{k-1}, \overline{P}_{c})-\epsilon _{k-1}\Phi(\epsilon _{k-2}, \overline{P}_{c})}{\Phi(\epsilon _{k-1}, \overline{P}_{c})-\Phi(\epsilon _{k-2}, \overline{P}_{c})}. 
\label{rec1}
\end{eqnarray}

Since the value of $\overline{P}_{c}$ is small, a solution can be obtained quickly by setting $\epsilon _{0}$ and $\epsilon _{1}$ close to zero. 
Then, we iterate until $\left|\epsilon_{k}-\epsilon_{k-1}\right|$ becomes very small, which is described in Algorithm 1. 
With the constraint rage $[0,\epsilon_{c}]$ fixed, we should be able to find an optimal energy detection threshold which maximizes the objective function $O(\epsilon)$. 

\begin{algorithm}
\caption{Optimization algorithm}\label{X}
\begin{algorithmic}[1]
\State Initialize $\epsilon _{0}$ and $\epsilon _{1}$.
\For{$k=1,2,3,...$}
\State Update $\epsilon _{k}$ using the recurrence relation (\ref{rec1}) 
\If{$\left |  \epsilon _{k}-\epsilon _{k-1}\right |$ is sufficiently small}
\State $\epsilon _{c}=\epsilon _{k}$
\Return $\epsilon _{c}$
\EndIf
\State \textbf{end if}
\EndFor
\State \textbf{end}
\If{$ \nabla O(\epsilon_{c})>0 $}
\State $\epsilon_{c}$ is the optimal solution.
\Else
\State Initialize $\epsilon _{0}=\epsilon_{c}$.
\For{$k=1,2,3,...$}
\State Update $\epsilon _{k}$ using the recurrence relation (\ref{rec2}) 
\If{$\left |  \epsilon _{k}-\epsilon _{k-1}\right |$ is sufficiently small}
\State $\epsilon^{*}=\epsilon _{k}$
\Return $\epsilon^{*}$
\EndIf
\State \textbf{end if}
\EndFor
\State \textbf{end} 
\EndIf
\State \textbf{end if}\end{algorithmic}
\end{algorithm}

We use the gradient descent method to find the maximum value in the constraint range with the recurrence relation as 
\begin{eqnarray}
\epsilon _{k}=\epsilon _{k-1}+\beta\!\:\nabla O(\epsilon_{k-1}) 
\label{rec2}
\end{eqnarray}
for the step size $\beta$. 
We use a fixed value for $\beta$ to avoid the complication of calculation and find an optimal one. 
In the first step of the recurrence, we assume $\epsilon _{0}=\epsilon _{c}$. 
From Proposition 1, if $\nabla O(\epsilon_{c})>0$, $\epsilon_{c}$ is an optimal point as the objective function is an increasing function in $[0,\epsilon_{c}]$. 
If $\nabla O(\epsilon_{c})<0$, we continue the recurrence process to find the maximum point given the objective function is a type-1 function with unique maximum point. 
Algorithm 1 for finding the optimal $\epsilon^{*}$ is stated above. 

Fig. \ref{fig:5} illustrates a whole process for subchannel clustering, spectral access and energy harvesting. The three traffic features are used to classify the PU traffic patterns through the BNP subchannel clustering. Based on the obtained idle and busy period statistics from the output of MCMC (Gibbs sampling), we find the corresponding subchannels $c_{t}$ and $c_{h}$ with maximum energy harvesting and transmission probabilities, respectively. If the residual energy in SU battery $B_{t}$ is less than the required energy for transmission, $c_{h}$ is selected to harvest energy and otherwise, $c_{t}$ to transmit data. Then, SU senses the selected subchannel using the optimal energy detection threshold obtained from the gradient descent method above. After sensing, if the sensing result is $o_{t}=0$, SU transmits data and otherwise, turns off the transmission. The residual energy in battery is then updated using (\ref{res}). To enable this, SU has to acquire information of $p ^ {c_ {h}} _ {i}$ and $p ^ {c_ {t}} _ {i}$ from the subchannel clustering, which in turn influences the energy harvesting and consuming rates $\rho_ {h}$ and $\rho_ {c}$, respectively. Therefore, the statistics information obtained through the accurate clustering process and resulting sensing parameter of the energy detection threshold play crucial role in obtaining the optimal sensing policy for SU.  

\begin{figure} [h]
\centering
\includegraphics[width=3.4in]{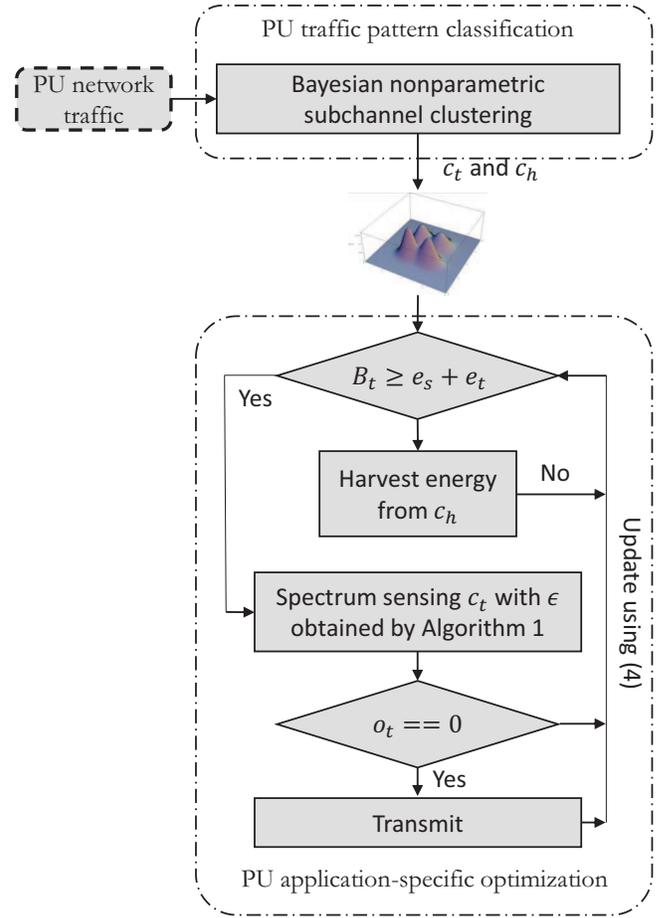}
\caption{A flow diagram for spectral access and energy harvesting.}
\label{fig:5}
\end{figure}

\section{Markov Decision Process Formulation}

Unlike the duty-cycling behavior between active mode and sleep mode, we introduce the MDP model to accurately predict the actions of SU which vary with the evolution of the residual energy in battery. Since idle/busy time distributions (to estimate harvesting and transmission opportunities) are obtained from clustering traffic patterns, we need to take decision in a real time. Even if we know the distribution parameters, we still need to automate the decision process because the SU harvesting and transmission is a real-time process. The parameters obtained would help us to achieve early convergence. For this purpose, the MDP which is Markov-chain based approach, is incorporated given idle/busy time distributions for traffic applications. Therefore, the MDP renders a good solution for online/real-time decision making.

\subsection{State Space of Battery}

SU decides whether to harvest energy or transmit based on the battery level. The event of $a_{t}=1$ means that the amount of the residual energy in battery is sufficient enough to transmit. To find the probability of having this event, we use the MDP formulation assuming the state is the discretization of the battery capacity, and evaluate the steady-state probabilities of the battery level with sufficient energy to transmit. We discretize the current residual energy in battery $b_{t}$ in $N_{b}$ levels where $N_{b}=\left\lfloor\frac{B_{max}}{e_{q}}\right\rfloor$ denotes the maximum amount of energy quanta that can be stored in battery. Here, one energy quantum corresponds to $e_{q}=\frac{(e_{s}+e_{t})}{n_{\tau}}$ where $n_{\tau}$ represents the number of states that enter harvesting mode. In general, if $N_{b}$ is sufficiently high, the discrete model can be considered as a good approximation of the continuous one.
Then, (4) can be rewritten in terms of energy quanta as 
\begin{eqnarray}
b_{t+1}=\min(b_{t}-e^{c}_{t}+e^{h}_{t},N_{b}) 
\label{qres}
\end{eqnarray}
where $e^{h}_{t}=\left\lfloor\frac{E^{h}_{t}}{e_{q}}\right\rfloor$ and $e^{c}_{t}=\left\lceil\frac{E^{c}_{t}}{e_{q}}\right\rceil$. Here, the floor is used to have a conservative harvesting performance, while the ceiling to assure a required energy consumption. Thus, the worst case of the battery level is assumed. 

\subsection{Transition Probability Matrix}

In the Markov chain model with $ N_ {b} $ states as shown in Fig. \ref{fig:6}, the harvesting state $i\in\left \{ 0, 1, ..., n_{\tau}-1 \right \}$ changes to state $j$ $(j\geq i)$ through energy harvesting as the current battery level is insufficient for transmission. The active state $i\in\left \{ n_{\tau}, n_{\tau}+1, ..., N_{b}-1 \right \}$ with sufficient energy to transmit will change to $ (i-n_{\tau})$-state or fail to transfer and return to $ i $-state again. 

\begin{figure} [h!]
\centering
\includegraphics[width=3.5in]{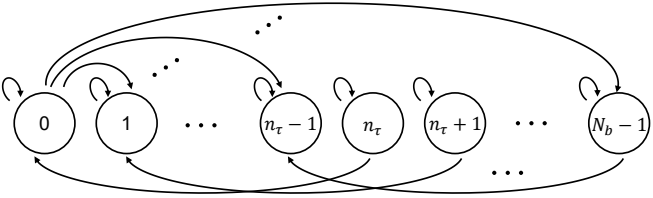}
\caption{The battery state transition with $n_{\tau}$ energy harvesting states for $N_{b}$-state Markov chain model.}
\label{fig:6}
\end{figure}

We define the transition probability matrix $\mathbf{U}$ for the $N_{b}$-state Markov chain model in Fig. \ref{fig:6} as
\begin{eqnarray}
\mathbf{U} (\epsilon|g)= \Big[\!\!
\begin{array}{ccc}
\mathbf{U}_{h} (\epsilon|g)  & \hspace{-0.1in}  | & \hspace{-0.1in} \mathbf{U}_{a} (\epsilon|g) \\
\end{array}  \!\!\Big]^{T} 
\end{eqnarray}
\begin{eqnarray}
\mathbf{U}_{h} (\epsilon|g) = \left(
\begin{array}{ccccccccc}
u_{0,0} & 0 & \dots & 0 \\
u_{1,0} & u_{1,1} & \dots & 0 \\
\vdots & \vdots & \ddots & 0 \\
u_{n_{\tau}-1,0} & \vdots & \ddots & u_{n_{\tau}-1,n_{\tau}-1} \\
u_{n_{\tau},0} & \vdots & \ddots & u_{n_{\tau},n_{\tau}-1} \\
u_{n_{\tau}+1,0} & \vdots & \ddots & u_{n_{\tau}+1,n_{\tau}-1} \\
\vdots & \vdots & \ddots & \vdots \\
u_{N_{b}-1,0} & \dots & \dots & u_{N_{b}-1,n_{\tau}-1} \\
\end{array} \right) 
\end{eqnarray}
\begin{eqnarray}
\mathbf{U}_{a} (\epsilon|g) = \hspace{2.3in}
\nonumber \\
\left(
\begin{array}{ccccccccc}
u_{0,n} & 0 & \dots & 0\\
0 & u_{1,n_{\tau}+1} & \dots & 0\\
0 & 0 & \ddots & 0\\
0 & 0 & 0 & u_{N_{b}-1-n_{\tau},N_{b}-1} \\
0 & 0 & 0 & 0\\
0 & 0 & 0 & 0 \\
u_{n_{\tau}-,n_{\tau}} & 0 & 0 & 0 \\
0 & u_{n_{\tau}+1,n_{\tau}+1} & 0 & 0 \\
0 & 0 & \ddots & 0 \\
0 & 0 & \dots & u_{N_{b}-1,N_{b}-1} \\
\end{array} \right) . 
\end{eqnarray}
Here, $\mathbf{U}_{h} (\epsilon|g)$ denotes the $N_{b}\times n_{\tau}$ matrix whose current battery level induces SU to enter harvesting mode, while $\mathbf{U}_{a} (\epsilon |g)$ the $N_{b}\times (N_{b}-n_{\tau})$ matrix whose current battery level induces SU to enter active mode. The components of these matrices are defined as 
\begin{align}
u_{i,i} &= 1_{\left[ 0 \leq E_{t}^{h} < e_{q} \right]}\, p^{c_{t}}_{o} +p^{c_{t}}_{i}, ~~(i<n_{\tau}) \\ 
u_{i,j} &= 1_{\left[ \left( i-j \right ) e_{q} \leq E_{t}^{h} < \left( i-j+1 \right)e_{q} \right]}\, p^{c_{t}}_{o}, \nonumber \\
&\mathrel{\phantom{=}} (0<j<n_{\tau}, ~j<i<N_{b}) \\ 
u_{i,i+n_{\tau}} &= \big[ 1-P_{f}\left( \epsilon \right) \big]\!\: p^{c_{t}}_{i} + \big[ 1-P_{d}\left( \epsilon |g \right) \big]\!\: p^{c_{t}}_{o}, ~~(i<n) \\ 
u_{i,i} &= P_{f} \left( \epsilon \right) p^{c_{t}}_{i} + P_{d} \left( \epsilon |g  \right) p^{c_{t}}_{o}, ~~(n\leq i\leq N_{b}-1). 
\end{align}
Here, $u_{i,i}$ $(i <n_{\tau})$ is the probability that energy harvesting is successful when $c_{h}$ is busy but there is insufficient energy to reach a higher battery level, or $c_{h}$ is idle. $u_{i,j}$ $(0<j<n_{\tau},~j<i<N_{b})$ is the one that energy harvesting is successful when $c_{h}$ is busy and the state changes from j to i. $u_{i,i+n}$ $(i<n_{\tau})$ is the one that $c_{t}$ is idle with no false alarm and successful transmission, or $c_{t}$ is busy with missed detection and collision. $u_{i,i}$ $(n \leq i \leq N_{b}-1)$ is the one that $c_{t}$ is idle with false alarm, or $c_{t}$ is busy with no missed detection.

\subsection{Steady-State Probability and Optimal Threshold Algorithm}

We define the steady-state probability vector of the $N_{b}$-state Markov chain as $\Pi=\left[ \pi_{0},\pi_{1},..., \pi_{N_{b}-1} \right]$, where $\Pi$ is the left eigenvector of $\mathbf{U} (\epsilon )$ corresponding to the unit eigenvalue as 
\begin{eqnarray}
\Pi\!\: \mathbf{U} (\epsilon|g)=\Pi .
\label{SSP}
\end{eqnarray}
To derive the steady-state probability vector $\Pi$, we need to make the necessary assumption below. 
\begin{assumption}
The maximum energy quanta $N_{b}$ should be sufficient enough to satisfy the following conditions:
\begin{subequations}
\begin{align}
e_{t}^{h} + n_{\tau}-1 &\leq N_{b}-1 \\
\left\lfloor\frac{E^{h}_{t}}{e_{q}}\right\rfloor &\leq N_{b}-n_{\tau} \\
\frac{E_{t}^{h}}{e_{q}} &<  N_{b}-n_{\beta}+1 \\
N_{b} &> \left( \frac{g \varphi P_{p}\!\: T_{t}}{e_{s}+e_{t}}+1 \right) n_{\tau}-1.
\end{align}
\end{subequations}
\end{assumption}
\noindent
Assumption 1 implies that the maximum state $N_{b}-1$ must be greater than the sum of the harvesting energy quanta $e_{t}^{h}$ and the maximum number of harvesting state $n_{\tau}-1$. This means that the maximum state number should always be greater than the maximum allowable state due to harvesting. Thus, (\ref{qres}) can be rewritten as 
\begin{eqnarray}
b_{t+1}=b_{t}-e^{c}_{t}+e^{h}_{t}.
\end{eqnarray}

With Assumption 1, we define the number of energy quanta charged through energy harvesting as 
\begin{align}
n_{\kappa}& = \left \{ n_{\kappa}\in \mathbb{N}\cap \left \{0\right \}\!\: \big|\, n_{\kappa}e_{q}\leq  E^{h}_{t}<(n_{\kappa}+1)e_{q} \right \}\nonumber \\
& = \left \lfloor  \frac{g n_{\tau} \varphi P_{p}T_{t}}{e_{s}+e_{t}} \right \rfloor. 
\end{align}
\begin{proposition}
Then, the steady-state probability vector $\Pi$ can be evaluated as 
\begin{equation}
\label{eqn:Ei}
\pi_{i}= 
\begin{cases}
\frac{\tau}{n_{\alpha}\alpha+n_{\tau}\tau}, & ~~(0\leq i<n_{\tau}) \\
\frac{\kappa}{n_{\kappa}\kappa+n_{\tau}\tau}, & ~~(n_{\tau}\leq i<n_{\tau}+n_{\kappa}) \\  
0, & ~~(n_{\tau}+n_{\kappa}\leq i<N_{b})
\end{cases}
\end{equation}
where
\begin{align}
\kappa &= p^{c_{h}}_{o} \\
\tau &= \big[ 1-P_{f}\left ( \epsilon \right ) \big]\!\: p^{c_{t}}_{i} + \big[ ( 1-P_{d}\left ( \epsilon |g \right ) \big]\!\: p^{c_{t}}_{o}.
\end{align}
\end{proposition}
\begin{proof}
See Appendix B.
\end{proof}

Using the steady-state probability vector $\Pi$ obtained by Proposition 3, the probability of SU entering active mode $P^{M}_{a}(\epsilon | g)$ can be derived as 
\begin{align}
P^{M}_{a}(\epsilon | g)=\sum_{i=n_{\tau}}^{N_{b}-1}\pi_{i} = \frac{n_{\kappa}\kappa}{n_{\kappa}\kappa+n_{\tau}\tau} = \frac{\big(\frac{n_{\kappa}}{n_{\tau}}\big)\kappa}{\big(\frac{n_{\kappa}}{n_{\tau}}\big)\kappa+\tau}. 
\end{align}
If $n_{\tau}$ is sufficiently large, we can approximate $\frac{n_{\kappa}}{n_{\tau}}\cong\frac{\varphi P_{p}T_{t}}{e_{s}+e_{t}}$ in (\ref{prob}), which yields 
\begin{equation}
P^{M}_{a}(\epsilon | g )=\frac{g \varphi P_{p}T_{t}\kappa}{g \varphi P_{p}T_{t}\kappa+(e_{s}+e_{t})\!\:\tau}. 
\end{equation}

To find the optimal energy detection threshold, we define the MDP objective function using $P^{M}_{a}(\epsilon | g)$ obtained by the MDP as $O^{M}(\epsilon)=(1-P_{f}\left ( \epsilon \right )) P^{M}_{a}(\epsilon | g)$, and then express the optimization problem again as 
\begin{eqnarray}
\max_{\epsilon }~ O^{M}\left ( \epsilon  | g\right )  ~~\textrm{s.t.}~~ P_{c}(\epsilon | g)\leq \overline{P}_{c}.
\end{eqnarray}
\begin{proposition}
The MDP objective function $O^{M}(\epsilon | g)$ can be classified into two types of function $f^{M}(\epsilon |g)$ as follows:
\begin{enumerate}
\item $f^{M}(\epsilon_{1}|g)<f^{M}(\epsilon_{2}|g)$ for $\epsilon_{1}<\epsilon_{2} \leq \epsilon_{m}$, \\
$f^{M}(\epsilon_{3}|g)>f^{M}(\epsilon_{4}|g)$ for $\epsilon_{m}\leq\epsilon_{3} < \epsilon_{4}$,\\
$\lim_{\epsilon\rightarrow \infty}f(\epsilon|g) =\gamma_{2}$ 
where
$\epsilon_{m}=\arg\max_{\epsilon}f^{M}(\epsilon|g)$.
\item $f^{M}(\epsilon_{1}|g)<f^{M}(\epsilon_{2}|g)$ for $\epsilon_{1}<\epsilon_{2}$,\\
$\lim_{\epsilon\rightarrow \infty}f^{M}(\epsilon|g) =\gamma_{2}$
\end{enumerate}
where
\begin{eqnarray}
\gamma_{2}=\frac{g \varphi P_{p}T_{t}\kappa}{g \varphi P_{p}T_{t}\kappa+e_{s}+e_{t}}.
\end{eqnarray}
\end{proposition}
\begin{proof}
See Appendix C.
\end{proof}
By Proposition 4 we can optimize in the same way as Algorithm 1 in Section IV. 

\section{Results}

The simulation results for the proposed duty cycle and MDP based stochastic models are also presented. To show the effectiveness of the proposed scheme by comparing how close to actual capacity, we define 
\begin{equation}
\beta = \frac{1}{N_{t}}\sum_{t=1}^{N_{t}}\frac{T_{t}}{T_{slot}}\,\boldsymbol{C}\!\: a_{t}(1-o_{t})(1-S_{c_{t}}) 
\label{acap}
\end{equation}
for the actual capacity obtained from simulation based on the Monte-Carlo method where $N_{t}$ is the number of simulation iterations using the exhaustive search for energy detection threshold. 
We use real wireless traces available online \cite{O1} \cite{O2} for 3G network. 
We utilize three sources (UDP, VoIP, Game) in our data set.
Unless otherwise stated, the values of the parameters used here are listed in Table \ref{table2}, which are mainly drawn from \cite {Y11}. 

\begin{table}[h!]
\caption{Simulation parameters}
\centering
\begin{tabular}{ |p{0.8cm}||p{5cm}|p{1.2cm}| }
\hline
Symbol & Definition & Value \\
\hline \hline
$W$ &Bandwith & 1MHz \\ \hline
$T_{slot}$&Slot duration & 10ms \\ \hline
$T_{s}$&Sensing duration & 2ms \\ \hline
$T_{t}$&Transmission duration & 98ms \\ \hline
$B_{0}$&Initial energy & 0 J \\ \hline
$P_{s}$&Sensing power & 110 mW \\ \hline 
$P_{t}$&Transmit power & 50 mW\\ \hline 
$P_{nc}$&Non-ideal circuit power & 115.9 mW\\ \hline 
$\eta$&Efficiency of power amplifier & -5.65 dB\\ \hline 
$\varphi$&Energy harvesting efficiency & 0.2\\ \hline 
$\sigma_{p}^{2}/\sigma_{w}^{2}$&SNR of PU signal at SU transmitter & -10 dB\\ \hline 
$\overline{P}_{c}$&Target probability of collision & 0.1\\ \hline 
$B_{max}$&Maximum capacity of battery & 1 mJ\\ \hline 
\end{tabular}\label{table2}
\end{table}

\subsection{Traffic-Awareness via Clustering}

Before evaluating the performance of the proposed scheme, we confirm the performance of traffic patterns clustering that results from the proposed MCMC based sampling and the variational inference. In Fig. \ref{fig:7}, we compare the accuracy of the MCMC based sampling, variational inference, and $K$-means (as baseline) clustering algorithms when data points (observations) are generated from 3 different (UDP, VoIP, Game) traffic patterns being mixed. The $K$-means algorithm is one for grouping a given data into $K$ clusters by minimizing the dispersion of the distance between each clusters. Unlike the other approaches, the $K$-means algorithm cannot estimate the number of clusters, and it should be performed only by assuming a fixed number of traffic sources. We see that the Bayesian approaches offer higher accuracy than the $K$-means algorithm. This is because the Bayesian approaches are to approximate a prior probability and a likelihood function derived from a statistical model for the observed data, whereas the $K$-means considers only the differences in observed traffic values. In the variational inference, we observe some errors compared to the MCMC method because we approximate the latent variables assumed by the mean-field theory. In Fig. \ref{fig:8}, we compare the two Bayesian approaches in terms of their elapsed times. We notice that the elapsed times increase as the number of data points increases, and it is confirmed that the variational inference shows less elapsed time than the MCMC method. Hence, if we can derive a set of equations used to iteratively update the parameters well, the former converges faster than the latter requiring a large amount of sampling work.

\begin{figure}
\centering
\includegraphics[width=3.5in]{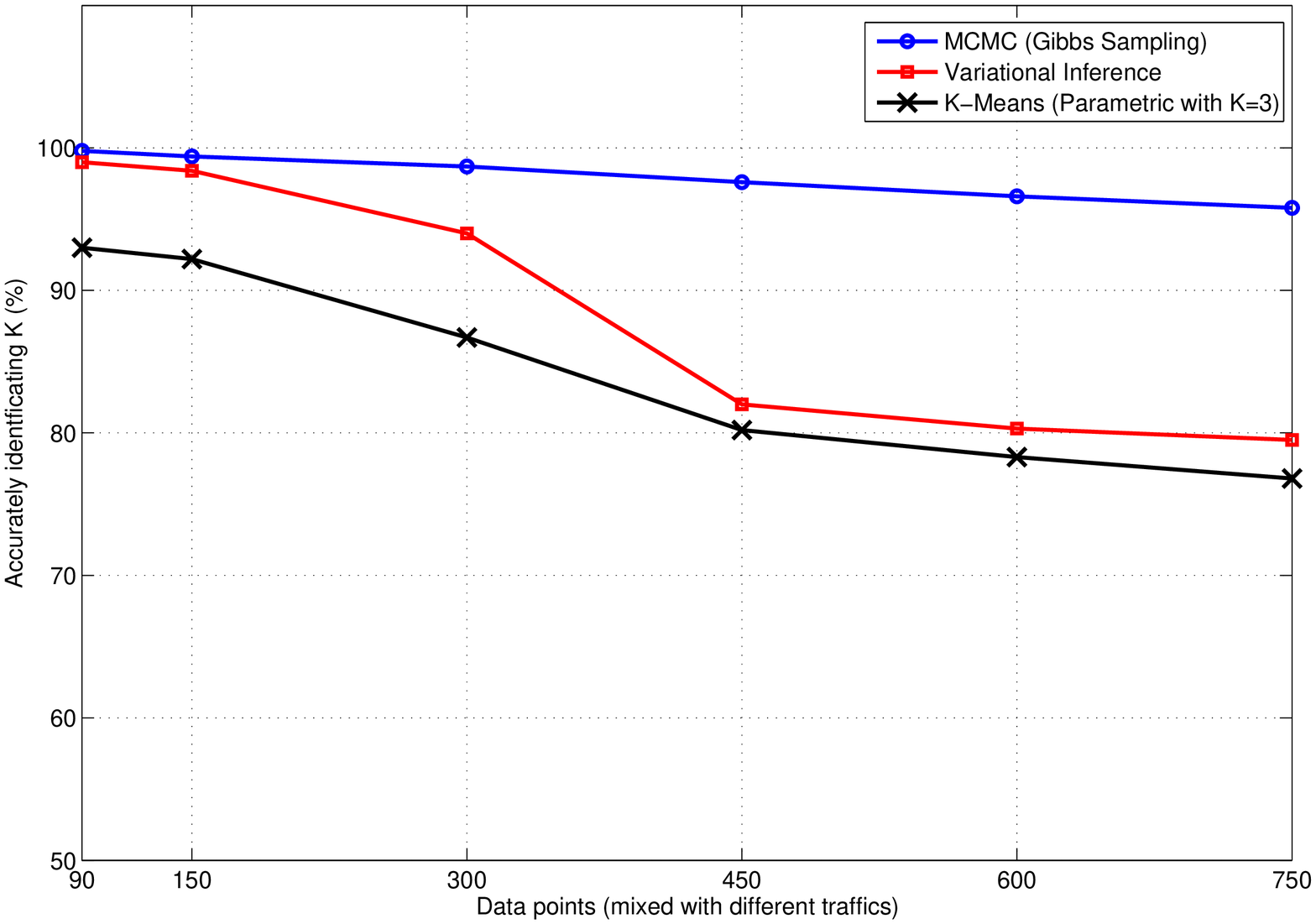}
\caption{Clustering accuracy of the proposed MCMC and variational inference, and $K$-means algorithms.}
\label{fig:7}
\end{figure}

\begin{figure}
\centering
\includegraphics[width=3.5in]{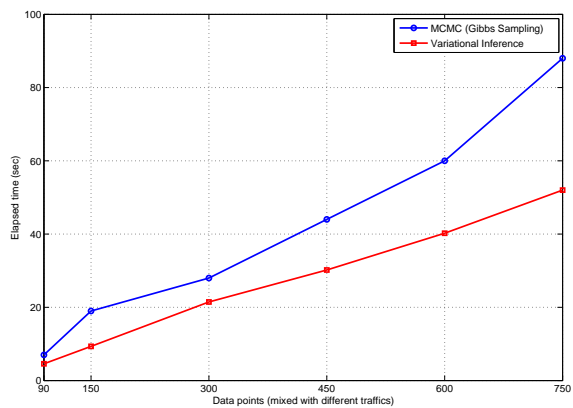}
\caption{Elapsed time of the MCMC and variational inference algorithms.}
\label{fig:8}
\end{figure}

\begin{figure}
\centering
\includegraphics[width=4.0in]{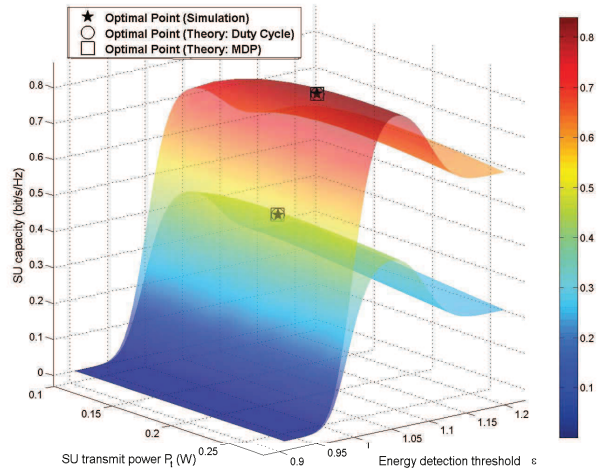}
\caption{SU capacity versus SU transmit power $P_{t}$ and energy detection threshold $\epsilon$. (top: VoIP with $p_{i}^{c_{h}}=0.2$ bottom: Game with $p_{i}^{c_{h}}=0.5$)}
\label{fig:9}
\end{figure} 

\begin{figure}
\centering
\includegraphics[width=3.6in]{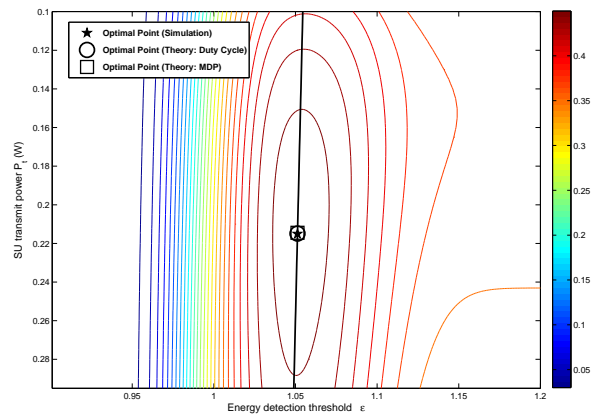}
\centerline{(a)}
\includegraphics[width=3.6in]{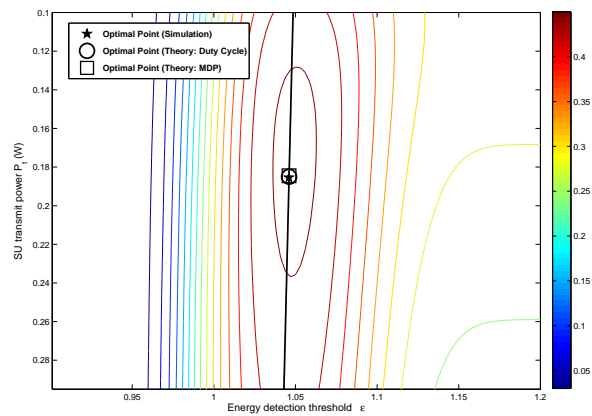}
\centerline{(b)}
\caption{An optimal point of SU transmit power $P_{t}$ and energy detection threshold $\epsilon$ for (VoIP, Game) traffic with ((a) VoIP $p_{i}^{c_{h}}=0.2$ (b) Game $p_{i}^{c_{h}}=0.5$)}
\label{fig:10}
\end{figure} 

\subsection{SU Optimal Sensing Threshold}

Fig. \ref{fig:9} shows the SU achievable rate capacity based on the duty cycle model with varying energy detection threshold $\epsilon$ and SU transmit power $P_{t}$. Note that the probability of the {\it harvest} subchannel $c_{h}$ being idle was measured as $p_{i}^{c_{h}}=(0.2, 0.5)$ for the (VoIP, Game) traffic applications, respectively. We observe an interesting trade-off in choice of the variables $\epsilon$ and $P_{t}$. For small $\epsilon$, the false alarm probability increases, resulting in low transmission probability. To the contrary, for large $\epsilon$, it decreases and SU more likely transmits data if its residual energy is enough for transmission. We notice that the capacity converges to a specific value proportional to $\gamma_{1}$ in (\ref{gam}), which is the ratio of the energy harvesting rate to the sum of the energy harvesting and consuming rates. Hence, an optimal $\epsilon$ balances the sensing accuracy trade-off. Likewise, there is the energy causality trade-off in the SU transmit power. For large $P_{t}$, the probability of SU being active decreases, while for small $P_{t}$, it increases but the SNR decreases. 
We see the VoIP traffic offering better performance than the Game traffic, as the former yields higher harvesting rate with continuous and short intervals between voice packets. However, the latter with $p_{i}^{c_{h}}=0.5$ shows the packet intervals changing dynamically, resulting in low harvesting rate. 

Fig. \ref{fig:10} shows an optimal point of the actual capacity in (\ref{acap}), the duty cycle and MDP models for (VoIP, Game) traffic. The black line shows an optimal $\epsilon$ obtained from (\ref{obj}) according to $P_{t}$, where the optimal $\epsilon$ decreases slightly as $P_{t}$ increases. This is because a tight energy causality due to the increased $P_{t}$ requires less transmission opportunity. We can see that the optimal point of VoIP traffic has a larger value of $P_{t}$ than that of Game traffic. In VoIP traffic, the increased $P_{t}$ results in low transmission opportunity but increases the SNR of SU, and a small value of $p_{i}^{c_{h}}$ guarantees the energy causality. It means that VoIP traffic subchannel idle/busy statistics, which show higher harvesting opportunity than Game traffic, offset the tight energy causality due to the increased $P_{t}$. The optimal point of the duty cycle and MDP models offer almost the same performance as the actual capacity from simulation. In Figs. \ref{fig:9} and \ref{fig:10}, it is evident that the subchannels carrying distinct traffic patterns exhibit different harvesting rates, and hence the appropriate values of $\epsilon$ should be determined considering both the energy causality and PU collision constraints. 

\subsection{SU Performance by Clustering Algorithms}

We evaluate the SU achievable rate capacity by the proposed clustering algorithms with respective threshold optimization. We assume 3 different (UDP, VoIP, Game) traffic sources with 10 subchannels, respectively. In this setting, SU selects the subchannels $c_{t}$ and $c_{h}$ from $N_{c}=30$ subchannels with 450 data points generated by using MCMC and variational inference, respectively. Then, the optimal energy detection threshold for the rate capacity is determined based on the duty cycle and MDP models. Fig. \ref{fig:11} shows the SU achievable rate capacity for varying energy detection threshold $\epsilon$. The black line represents actual capacity using $\epsilon$ obtained from (\ref{obj}) with accurate clustering information. We see that the MCMC is closer to the optimal line than variational inference, while the duty cycle and MDP models offer almost the same performance in optimizing $\epsilon$. This clearly shows the higher sensing accuracy in selecting $c_{t}$ and $c_{h}$ subchannels of the MCMC than variational inference. 

Fig. \ref{fig:12} shows the SU achievable rate capacity obtained by using the optimal value $\epsilon^{*}$ for varying SU transmit power $P_{t}$. We see that the variational inference reaches maximum when $P_{t}=0.24W$, but the optimal actual capacity does slightly later when $P_{t}=0.26W$, like the MCMC. It means that accurate clustering information leads to higher energy harvesting rate, which allows SU to increase the residual energy in battery. Hence, SU can increase the maximum achievable rate capacity with higher transmit power $P_{t}$.

\begin{figure}
\centering
\includegraphics[width=3.5in]{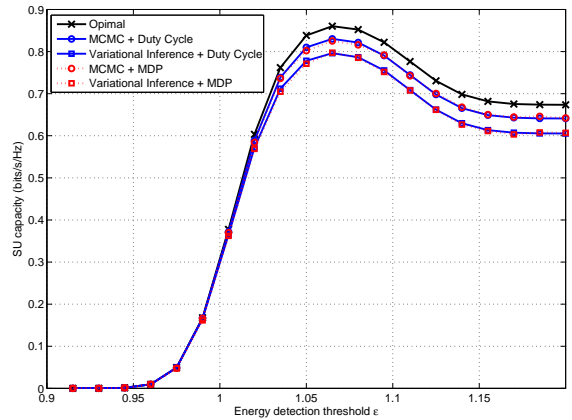}
\caption{SU rate capacity versus energy detection threshold $\epsilon$ when $P_{t}$=0.2W.}
\label{fig:11}
\end{figure} 

\begin{figure}
\centering
\includegraphics[width=3.5in]{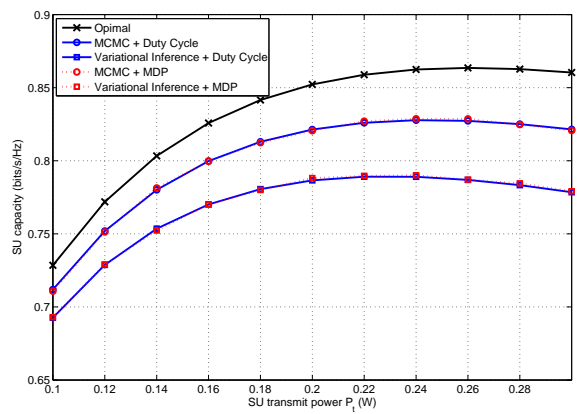}
\caption{SU rate capacity versus SU transmit power $P_{t}$ with optimal $\epsilon^{*}$.}
\label{fig:12}
\end{figure} 

\section{Conclusion}

We have proposed an optimal spectrum sensing policy for maximizing the SU capacity in OFDMA based CRNs which is powered by energy harvesting.
SU collected traffic pattern information through observation of PU subchannels and classified the idle/busy period statistics for each PU subchannel using the MCMC and variational inference algorithms. Based on these statistics, we developed the stochastic SU capacity models which are the duty cycle based one defined by the times spent in active and sleep mode, and the MDP model based on the evolution of the residual energy in battery. The energy detection threshold was optimized to maximize the SU capacity while satisfying the energy causality and PU collision constraints according to traffic patterns. We have shown the performance trade-off of the BNP subchannel clustering algorithms by comparing the accuracy and elapsed time of algorithms. It was shown that SU can optimize the stochastic model by selecting the threshold referring to the idle/busy period statistics of PU subchannels. It was also shown that the proposed duty cycle and MDP model achieve similar capacity to that of the actual capacity from simulation based on the Monte-Carlo method.

\section*{Acknowledgment}

This work was supported by the National Research Foundation of Korea (NRF) Grant funded by the Korean Government under Grant 2014R1A5A1011478.

\appendices
\section{Proof of proposition 1}
We define $A(\epsilon|g)=[1-P_{d}\left ( \epsilon|g \right )]$ and $B(\epsilon|g)=P_{a}\left ( \epsilon |g \right )$ for the proof of the collision probability $P_{c}(\epsilon |g)$ in (\ref{con}).
\begin{eqnarray}
\frac{\mathrm{d}f(\epsilon|g) }{\mathrm{d} \epsilon}=\frac{\mathrm{d} A(\epsilon|g) }{\mathrm{d} \epsilon }B(\epsilon|g)+A(\epsilon|g)\frac{\mathrm{d} B(\epsilon|g) }{\mathrm{d} \epsilon}.
\end{eqnarray}
The increasing function $A(\epsilon|g)$ and the decreasing function $B(\epsilon|g )$ satisfy the following properties: 
\begin{align}
\lim_{\epsilon\rightarrow 0}A(\epsilon|g) &< \lim_{\epsilon\rightarrow 0}B(\epsilon|g)
\label{pro1}
\\
\lim_{\epsilon\rightarrow \infty }A(\epsilon|g) &> \lim_{\epsilon\rightarrow \infty}B(\epsilon|g)
\label{pro2}
\\
A(\epsilon_{e}|g)&=B(\epsilon_{e}|g)
\label{pro3}
\\
\frac{\mathrm{d} A(\epsilon|g) }{\mathrm{d} \epsilon }&>-\frac{\mathrm{d} B(\epsilon|g) }{\mathrm{d} \epsilon}.
\label{pro4}
\end{align}
The same results can be obtained even if $A(\epsilon)=[1-P_{f}\left ( \epsilon \right )]$ for the proof of the objective function $O(\epsilon |g)$ in (\ref{obj}).
\begin{proposition}
If $0<\epsilon<\epsilon_{e}$, $\frac{\mathrm{d}f(\epsilon|g) }{\mathrm{d} \epsilon}>0$.
\end{proposition}
\begin{proof}
We have $A(\epsilon|g)<B(\epsilon|g)$ from (\ref{pro1}) and (\ref{pro3}), and using (\ref{pro4}), we have the following inequalities:
\begin{subequations}
\begin{align}
\frac{\mathrm{d} A(\epsilon|g) }{\mathrm{d} \epsilon }B(\epsilon|g)&>B(\epsilon|g)\left ( -\frac{\mathrm{d} B(\epsilon|g) }{\mathrm{d} \epsilon }\right ) 
\\
\frac{\mathrm{d} A(\epsilon|g) }{\mathrm{d} \epsilon }B(\epsilon|g)&>A(\epsilon|g)\left ( -\frac{\mathrm{d} B(\epsilon|g) }{\mathrm{d} \epsilon }\right ) 
\\
\frac{\mathrm{d}f(\epsilon|g) }{\mathrm{d} \epsilon}&>0. 
\end{align}
\end{subequations}
\end{proof}
\begin{proposition}
If $\epsilon_{e}<\epsilon$ and $\frac{\mathrm{d}f(\epsilon|g) }{\mathrm{d} \epsilon}<0$, $\frac{\mathrm{d}f(\epsilon+\Delta |g) }{\mathrm{d} \epsilon}<0$ where $\Delta\in \mathbb{R}^{+}$.
\end{proposition}
\begin{proof}

\begin{subequations}
\begin{align}
\frac{\mathrm{d} A(\epsilon+\Delta|g) }{\mathrm{d} \epsilon }<\frac{\mathrm{d} A(\epsilon|g) }{\mathrm{d} \epsilon}  \label{pr61a}
\\
\frac{\frac{\mathrm{d} A(\epsilon+\Delta|g) }{\mathrm{d} \epsilon }}{A(\epsilon|g)}<\frac{\frac{\mathrm{d} A(\epsilon|g) }{\mathrm{d} \epsilon }}{A(\epsilon|g)}\label{pr61b}
\\
\frac{\frac{\mathrm{d} A(\epsilon+\Delta|g) }{\mathrm{d} \epsilon }}{A(\epsilon+\Delta|g)}<\frac{\frac{\mathrm{d} A(\epsilon|g) }{\mathrm{d} \epsilon}}{A(\epsilon|g)}. \label{pr61c}
\end{align}
\end{subequations}
Inequality (\ref{pr61a}) is satisfied by $\epsilon_{i}<\epsilon_{e}$ where $\epsilon_{i}$ is the inflection point of $A(\epsilon|g)$, and (\ref{pr61c}) is satisfied by the increasing function property $A(\epsilon+\Delta|g)>A(\epsilon|g)$
\begin{subequations}
\begin{align}
\frac{\frac{\mathrm{d} A(\epsilon|g) }{\mathrm{d} \epsilon }}{A(\epsilon|g)}&<\frac{-\frac{\mathrm{d} B(\epsilon+\Delta|g) }{\mathrm{d} \epsilon }}{B(\epsilon|g)} \label{pr62a}
\\
\frac{\frac{\mathrm{d} B(\epsilon+\Delta|g) }{\mathrm{d} \epsilon }}{B(\epsilon|g)}&<\frac{-\frac{\mathrm{d} B(\epsilon+\Delta|g) }{\mathrm{d} \epsilon }}{B(\epsilon+\Delta|g)} \label{pr62b}
\\
\frac{\frac{\mathrm{d} A(\epsilon+\Delta|g) }{\mathrm{d} \epsilon }}{A(\epsilon+\Delta|g)}&<\frac{-\frac{\mathrm{d} B(\epsilon+\Delta|g) }{\mathrm{d} \epsilon }}{B(\epsilon+\Delta|g)}. \label{pr62c}
\end{align}
\end{subequations}
Inequality (\ref{pr62a}) is satisfied by $\epsilon_{e}<\epsilon$, and (\ref{pr62b}) is satisfied by the decreasing function property $B(\epsilon+\Delta|g)<B(\epsilon|g)$. Finally, from (\ref{pr61c}), the inequality (\ref{pr62c}) holds.
\end{proof}
From Propositions 5 and 6, the equation $\frac{\mathrm{d}f(\epsilon|g) }{\mathrm{d} \epsilon}=0$ has one solution or no solution, so that the function $f(\epsilon|g)$ converges to $\gamma_{1}=(\lim_{\epsilon\rightarrow 0}A(\epsilon|g))(\lim_{\epsilon\rightarrow 0}B(\epsilon|g) ) $. 
\
\section{Proof of proposition 3}
From (\ref{SSP}), we have
\begin{subequations}
\begin{align}
\Pi\left ( \mathbf{I}-\mathbf{U}(\epsilon |g) \right ) &=0
\\
\left ( \mathbf{I}-\mathbf{U}(\epsilon |g) \right )^{T}\Pi^{T} &=0
\\
\mathbf{G}(\epsilon |g)\Pi^{T} &=0.
\end{align}
\end{subequations}
We define $\mathbf{G}(\epsilon |g)=\left ( \mathbf{I}-\mathbf{U}(\epsilon |g) \right )^{T}$. If Assumption 1 holds, then $\mathbf{G}(\epsilon |g)$ is given by
\begin{eqnarray}
\mathbf{G} (\epsilon |g)= \Big[\!\!
\begin{array}{ccc}
\mathbf{G}_{1} (\epsilon |g) & \hspace{-0.1in}  \Big| & \hspace{-0.1in} \mathbf{G}_{2} (\epsilon |g) \\
\end{array} \!\!\Big] 
\end{eqnarray}
\begin{eqnarray}
\mathbf{G}_{1} (\epsilon  |g) = \left(
\begin{array}{cccccccccccccccc}
\alpha &  &   \\
 &  \ddots  & \\
 &  &  \alpha \\
0 & 0 & 0  \\
 & \vdots &  \\
0 & 0 & 0  \\
-\alpha &  &   \\
 &  \ddots  & \\
 &  &  -\alpha \\
0 & 0 & 0  \\
 & \vdots &  \\
\end{array} \right)  
\end{eqnarray}
\begin{eqnarray}
\mathbf{G}_{2} (\epsilon  |g) = \left(
\begin{array}{ccccccccc}
-\beta &  &   \\
 &  \ddots  & \\
 &  &  -\beta \\
0 & 0 & 0  \\
 & \vdots &  \\
0 & 0 & 0  \\
\beta &  &   \\
 &  \ddots  & \\
 &  &  \beta \\
\end{array} \right) . 
\end{eqnarray}

The steady-state vector $\Pi^{T}$ is the null vector of $\mathbf{U} (\epsilon |g)$, and hence the vector in (\ref{eqn:Ei}) is a kind of the null vector.
\
\section{Proof of proposition 4}
We define $A(\epsilon )=(1-P_{f}\left ( \epsilon \right ))$ and $B(\epsilon|g)=P_{a}\left ( \epsilon |g \right )$ for the proof of the MDP objective function $O^{M}(\epsilon |g)$ case.
\begin{eqnarray}
\frac{\mathrm{d}f^{M}(\epsilon|g) }{\mathrm{d} \epsilon}=\frac{\mathrm{d} A(\epsilon) }{\mathrm{d} \epsilon }B(\epsilon|g)+A(\epsilon)\frac{\mathrm{d} B(\epsilon|g) }{\mathrm{d} \epsilon}. 
\end{eqnarray}
The increasing function $A(\epsilon )$ and the decreasing function $B(\epsilon | g)$ satisfy the following properties (\ref{pro1}) - (\ref{pro4}). From Proposition 5, we can derive that if $0<\epsilon<\epsilon_{e}$, $\frac{\mathrm{d}f^{M}(\epsilon|g) }{\mathrm{d} \epsilon}>0$.
Also, from Proposition 6, we can derive that if $\epsilon_{e}<\epsilon$ and $\frac{\mathrm{d}f^{M}(\epsilon|g) }{\mathrm{d} \epsilon}<0$, $\frac{\mathrm{d}f^{M}(\epsilon+\Delta |g) }{\mathrm{d} \epsilon}<0$ where $\Delta \in \mathbb{R}^{+}$. 
Thus, the equation $\frac{\mathrm{d}f^{M}(\epsilon| g) }{\mathrm{d} \epsilon}=0$ has one solution or no solution, and the function $f^{M}(\epsilon| g)$ converges to $\gamma_{2}=(\lim_{\epsilon\rightarrow 0}A(\epsilon))(\lim_{\epsilon\rightarrow 0}B(\epsilon| g) ) $.


\begin{thebibliography}{1}

\bibitem{T11} T. Chen, Y. Yang, H. Zhang, H. Kim, and K. Horneman, ``Network energy saving technologies for green wireless access networks,"
{\em IEEE Wireless Commun.}, vol. 18, no. 5, pp. 30-38, Oct. 2011.

\bibitem{C11} C. Han et al., ``Green radio: Radio techniques to enable energy-efficient wireless networks,"
{\em IEEE Commun. Mag.}, vol. 49, no. 6, pp. 46-54, Jun. 2011.

\bibitem{X15} X. Lu, P. Wang, D. Niyato, D. I. Kim, and Z. Han, ``Wireless networks with RF energy harvesting: A contemporary survey," 
{\em IEEE Commun. Surveys \& Tutorials}, vol. 17, no. 2, pp. 757-789, Second Quarter 2015.

\bibitem{X16} X. Lu, P. Wang, D. Niyato, D. I. Kim, and Z. Han, ``Wireless charging technologies: Fundamentals, standards, and network applications," 
{\em IEEE Commun. Surveys \& Tutorials}, vol. 18, no. 2, pp. 1413-1452, Second Quarter 2016. 

\bibitem{A12} A. M. Zungeru, L. Ang, S. Prabaharan, and K. P. Seng, ``Radio frequency energy harvesting and management for wireless sensor networks," 
{\em Green Mobile Devices and Networks: Energy Optimization and Scavenging Techniques}, ch. 13, pp. 341-368, CRC Press 2012. 

\bibitem{L13} S. Lee, R. Zhang, and K. Huang, ``Opportunistic wireless energy harvesting in cognitive radio networks,'' 
{\em IEEE Trans. Wireless Commun.}, vol. 12, pp. 4788-4799, Sep. 2013.

\bibitem{S13} S. Park, H. Kim, and D. Hong, ``Cognitive radio networks with energy harvesting," 
{\em IEEE Trans. Wireless Commun.}, vol. 12, pp. 1386-1397, Mar. 2013.

\bibitem{Q07} Q. Zhao, L. Tong, A. Swami, and Y. Chen, ``Decentralized cognitive MAC for opportunistic spectrum access in ad hoc networks: A POMDP framework,"
{\em IEEE J. Select. Area. Commun.}, vol. 25, no. 3, pp. 589-600, 2007. 

\bibitem{Blei} D. Blei and M. Jordan, ``Variational inference for Dirichlet process mixtures,''
{\em Bayesian Analysis}, vol. 1, no. 1, pp. 121-144, Aug. 2006.

\bibitem{Wood} F. Wood and M. J. Black, ``A nonparametric bayesian alternative to spike sorting,''
{\em Journal of Neuroscience Methods}, vol. 173, no. 1, pp. 1-12, Jun. 2008.

\bibitem{GrGh05} T. L. Griffiths and Z. Ghahramani, ``Infinite latent feature models and the Indian buffet process,''
Gatsby Computational Neuroscience Unit, Tech. Rep. 2005-001, 2005.

\bibitem{S12} S. Park, J. Heo, B. Kim, W. Chung, H. Wang, and D. Hong, ``Optimal mode selection for cognitive radio sensor networks with RF energy harvesting," {\em IEEE Proc. PIMRC 2012}, pp. 2155-2159, 2012.

\bibitem{S13} S. Park, H. Kim, and D. Hong, ``Cognitive radio networks with energy harvesting,"
{\em IEEE Trans. Wireless Commun.}, vol. 12, no. 3, pp. 13861397, Mar. 2013.

\bibitem{Bishop_book} C. M. Bishop,
``Mixtures of gaussians,'' in {\em Pattern recognition and machine learning} 1st ed., Springer-Verlag New York, ch. 9, sec. 2, pp. 430-435.

\bibitem{AS12} A. Sultan, ``Sensing and transmit energy optimization for an energy harvesting cognitive radio,"
{\em IEEE Wireless Commun. Letters}, vol. 1, no. 5, pp. 500-503, Oct. 2012.

\bibitem{D14} D. T. Hoang, D. Niyato, P. Wang, D. I. Kim, ``Opportunistic channel access and RF energy harvesting in cognitive radio network,'' {\em IEEE Journal on Selected Areas in Communications - Cognitive Radio Series}, vol. 32, pp. 2039-2052, Nov. 2014.

\bibitem{D15} D. T. Hoang, D. Niyato, P. Wang, D. I. Kim, ``Performance optimization for cooperative multiuser cognitive radio networks with RF energy harvesting capability,'' {\em IEEE Trans. Wireless Commun.}, vol. 14, pp. 3614-3629, July 2015.

\bibitem{M14} M. E. Ahmed, J. B. Song, Z. Han, and D. Y. Suh, ``Sensing-transmission edifice using Bayesian nonparametric traffic clustering in cognitive radio networks," 
{\em IEEE Trans. Mobile Computing}, vol. 13, pp. 2141-2155, Sep. 2014. 

\bibitem{Y16} M. E. Ahmed, D. I. Kim, J. Y. Kim, and Y. A. Shin, ``Energy-arrival-aware detection threshold in wireless-powered cognitive radio networks," 
{\em IEEE Trans. Vehic. Technol.}, vol. 66, pp. 9201-9213, Oct. 2017. 

\bibitem{K16} M. E. Ahmed, D. I. Kim, and K. W. Choi, ``Traffic-aware optimal spectral Access in wireless powered cognitive radio networks," 
{\em IEEE Trans. Mobile Computing}, vol. 17, pp. 734-745, Mar. 2018. 

\bibitem{J14} J. Xu, and R. Zhang, ``Throughput optimal policies for energy harvesting wireless transmitters with non-ideal circuit power,"
{\em IEEE J. Select. Area. Commun.}, vol. 32, pp. 322-332, Feb. 2014.

\bibitem{M12} M. E. Ahmed, J. B. Song, N. T. Nguyen, and Z. Han, ``Nonparametric Bayesian identification of primary users' payloads in cognitive radio networks,'' {\em IEEE Proc. ICC 2012,} pp. 1586-1591, June 2012.

\bibitem{ME14} M. E. Ahmed, J. B. Song, Z. Han, and D. Y. Suh, ``Sensing-transmission edifice using Bayesian nonparametric traffic clustering in cognitive radio networks,'' {\em IEEE Trans. Mobile Computing,} vol. 13, no. 9, pp. 2141-2155, Sept. 2014.

\bibitem{W14} W. S. Chung, S. S. Park, S. M. Lim, and D. S. Hong, ``Spectrum sensing optimization for energy-harvesting cognitive radio systems," 
{\em IEEE Trans. Wireless Commun.}, vol. 13, pp. 2601-2613, May 2014.

\bibitem{Y11} Y. Pei, Y. C. Liang, K. C. Teh, and K. H. Li, ``Energy-efficient design of sequential channel sensing in cognitive radio networks: Optimal sensing strategy, power allocation, and sensing order," 
{\em IEEE J. Select. Area. Commun.}, vol. 29, pp. 1648-1659, Aug. 2011. 

\bibitem{O1} [Online]. Available FTP: \url{http://crawdad.cs.dartmouth.edu/meta.php?name=snu/wow via wimax}

\bibitem{O2} [Online]. Available FTP: \url{http://crawdad.cs.dartmouth.edu/meta.php?name=kaist/wibro}

\end{thebibliography}
\end{document}